%% file: main.tex
\documentclass[preprint,12pt]{elsarticle}

\usepackage{amsmath, amssymb}
\usepackage{amsthm}
\usepackage{graphicx}
\usepackage{booktabs}
\usepackage{adjustbox}
\usepackage{float}
\usepackage{placeins}
\usepackage{hyperref}

\newtheorem{lemma}{Lemma}
\newtheorem{theorem}{Theorem}

\newtheorem{proposition}{Proposition}

\makeatletter\def\ps@pprintTitle{\let\@oddhead\@empty\let\@evenhead\@empty\let\@oddfoot\@empty\let\@evenfoot\@empty}\makeatother

\begin{document}

\begin{frontmatter}

\title{Learned Monotone Recurrent Features in Governed Credit Scoring: The Price of the Frame and the Necessity of Macro Conditioning}
\author{Yew Lee Tan\fnref{fn1}}
\ead{yewlee001@e.ntu.edu.sg}
\fntext[fn1]{This work was conducted in a personal capacity; the views expressed are the author's own.}

\begin{abstract}
Regulated credit scoring requires scores monotone non-decreasing in every exposure input. Deployed pipelines---hand-crafted monotone aggregates feeding sign-constrained gradient boosting---already meet this by composition; the open question is what \emph{learned} temporal aggregation is worth inside one. We answer on five production-scale credit datasets at matched admissibility (one priced baseline convention excepted), with a monotone recurrent architecture whose per-input guarantee we extend, with proofs, to vector-valued inputs and to exogenously macro-conditioned decay gates, severities, thresholds, and peak memory. Two findings result. First, a \emph{strictness ladder}: the value of learned monotone features rises with governance-frame strictness---zero on unconstrained engineered panels, maximal in summaries-only frames---replicated across two datasets and an official temporal-stability metric, though unconditioned features degrade on externally adjudicated later weeks. Second, a \emph{conditioning--delivery asymmetry} under regime shift. On a train-on-boom, test-on-crisis mortgage design, two public macroeconomic series hurt as input columns, yet conditioning the recurrence on them delivers the paper's only learned-block crisis-cohort uplifts. The confirmed effect: $+0.006$ to $+0.013$ AUC on an internally pre-registered Freddie Mac replication, at all five held-out seeds. The discovery estimate: $+0.015$ to $+0.021$ on Fannie Mae (three of five seeds post hoc), worth $10$--$27$ basis points of defaulted balance at an $80\%$ approval cutoff, and grows with early-prepaid loans excluded. A state-level test identifies the mechanism: between-cohort calibration transfer. A pandemic-band episode bounds scope: under forbearance-distorted labels the gain generalizes at a quarter to a third of crisis size on Fannie Mae, on Freddie Mac only against the capacity control.
\end{abstract}


\end{frontmatter}

\section{Introduction}
\label{sec:intro}

Exposure time series arise throughout consumer and mortgage credit---utilization, balances, delinquency depth, or limit usage over time---and larger exposures are expected to correspond to greater risk. In regulated risk scoring this expectation is formalized as a monotonicity constraint: increasing any exposure input must never decrease the resulting score. Supervisory guidance requires risk models to be conceptually sound, explainable, and verifiable through effective challenge \citep{fed2011_sr11_7,occ2011_12}, and a monotone exposure--score relationship is a standard way to meet that expectation. The constraint is normative rather than statistical. A model that violates it is inadmissible regardless of its predictive accuracy, so it must hold exactly, for every input, not merely on average over a validation set.

Deployed practice already meets this constraint, but through \emph{hand-crafted} features. Credit pipelines collapse exposure histories into windowed aggregates---trailing means and maxima, breach counts, last values---of the kind packaged as bureau trended-data attributes \citep{experian_trended3d,vantagescore4_userguide}, and feed them to a scorecard model, increasingly a gradient-boosted ensemble with per-feature monotone constraints \citep{chen2016xgboost,ke2017lightgbm,lightgbm_monotone}. Each such aggregate is monotone in every exposure input, and the constrained ensemble is monotone in every feature, so the pipeline is monotone end-to-end by composition. The guarantee, by itself, is available to engineered features today. What the engineered pipeline cannot do is \emph{adapt} its temporal aggregation: every window length, decay, and threshold is fixed ex ante by an analyst.

This paper asks the two questions that deployment forces. First, what is learned monotone temporal aggregation worth inside a governed pipeline---and does its value depend on how strictly the surrounding feature frame is governed? Second, does the learned aggregation survive a genuine macroeconomic regime shift, and if it does not, what repairs it?

We answer on five production-scale credit datasets using a measurement instrument developed in a companion study \citep{instrument2026}: a recurrent architecture whose hidden state consists of classical risk statistics---a discounted exposure level (an exponentially weighted moving average, EWMA) and a high-water mark (HWM)---and whose sequence-to-scalar score is monotone non-decreasing in every exposure input by construction, a guarantee that survives stochastic inference. Because the engineered incumbent is monotone end-to-end by the same composition, every comparison in this paper runs at \emph{matched admissibility}---any measured difference is the value of learning the aggregation, not the value of the constraint---up to one disclosed baseline convention (a sign-free volatility aggregate inside the supervised summary scalars) that strengthens the baseline and is priced in Section~\ref{sec:real-setup}. The companion study validates the instrument on a controlled synthetic suite with known monotone mechanisms---zero violations, with the constraint's cost bounded below $0.0017$ AUC---and establishes the mechanism behind our panel results: by direct functional regression, the learned channel spans the geometrically weighted separable family of hand-crafted statistics, failing only at consecutivity, time localization, and the exposure floor (statistics requiring consecutive runs of breaches, weight profiles tied to specific positions in the sequence, and the running minimum of exposure). The division of contribution is thus: the companion supplies the instrument, its synthetic validation, and its expressiveness boundary; this paper supplies the two extensions with their proofs (Section~\ref{sec:method}) and all real-data measurement.

\paragraph{Finding 1: the strictness ladder.}
The accuracy value of learned monotone recurrent features is a monotone function of the strictness of the governance frame that consumes them. The uplift is statistically zero when the downstream model is free to use a large unconstrained engineered panel. It is boundary-level at fifteen pre-engineered channels when every retained feature must carry a signed monotonicity commitment (intervals graze zero; reported, not claimed) and zero at thirteen raw ones. It is maximal---confidence intervals clear of zero in every seed pair on the in-distribution regimes---when only auditable summaries are admitted. The pattern replicates across two datasets and an official temporal-stability metric. The uplift requires channel breadth---absent from two through eight channels, clear from thirteen up---and is an order of magnitude larger at thirteen raw event channels than at fifteen pre-engineered ones, a gap confounded with channel rawness and metric scale, so the clean breadth evidence is the low-breadth absence.

\paragraph{Finding 2: conditioning beats delivery under regime shift.}
On a train-on-boom, test-on-crisis mortgage design (a $7{\times}$ default-rate swing), unconditioned learned recurrence does not survive: the paired contrast changes sign with the seed, and neither seed-ensembling nor doubled channel breadth revives it. On externally adjudicated later weeks of the Home Credit Stability benchmark, the unconditioned features' temporal-stability contribution likewise degrades outright (Section~\ref{sec:real-hcs}). Handing the downstream model two public macroeconomic series directly as columns \emph{hurts} under the shift (CI-clear negative; the term is defined in Section~\ref{sec:real-setup}). The same two series, conditioning the recurrence structurally through gates, severities, and thresholds, deliver a $+0.015$ to $+0.021$ AUC crisis-cohort uplift at every seed under the pre-registered ensembled protocol (\emph{pre-registered} throughout means internally frozen before the corresponding fits ran; Section~\ref{sec:real-setup}), beating a matched static bank at every one (three of the five seeds are a post-hoc robustness expansion). In distribution the ordering inverts. A pre-registered state-level geography test bounds the mechanism: the gain is between-cohort calibration transfer---exactly what an exogenous regime handle should buy, and something post-hoc recalibration cannot supply, since recalibration needs crisis labels after the fact while conditioning transfers before them. A pre-registered replication on a second GSE portfolio (Freddie Mac; same crisis, independent loan population and data pipeline) confirms the conditioning uplift at all five RNN seeds under the strictest decision bar used in this paper---$+0.006$ to $+0.013$ AUC, the paper's confirmed effect size, with the larger Fannie range read as the discovery estimate. A pre-registered pandemic-band study (2020--2021 vintages, forbearance-distorted labels, caveats declared before any fit) then bounds episode generality: on Fannie Mae the same two contrasts clear at all five seeds at a quarter to a third of the crisis magnitude; on Freddie Mac only the conditioning-versus-capacity contrast clears---the episode premium is crisis-specific (Section~\ref{sec:fannie-pandemic}).

In one sentence: the value of learned monotone recurrence equals the price of the freedom the governance frame removes---and under temporal regime shift it must be exogenously conditioned, or its temporal contribution should not be trusted beyond the training window.

\paragraph{Contributions.}
\begin{itemize}
  \item \textbf{The strictness ladder} (Sections~\ref{sec:real-strictness}, \ref{sec:real-hcs}): the first measurement, to our knowledge, of learned-monotone-feature value as a function of governance-frame strictness, replicated across two datasets and an official temporal-stability metric, with the frame's own accuracy price measured alongside (the top rung buys back about five percent of the summaries frame's surrendered accuracy).
  \item \textbf{The conditioning--delivery asymmetry under regime shift} (Section~\ref{sec:real-fannie}): a demonstration that exogenous macro series help a governed scorecard under crisis shift \emph{only} as structural conditioning of the recurrence---as free columns they hurt (exploratory at the national level, pre-registered and replicated in the state-granularity test)---together with a pre-registered state-level test establishing the mechanism as between-cohort calibration transfer, an exploratory content control attributing the gain to the macro content itself (a calendar-drift handle through the identical bank is CI-clear negative under the shift), scope cells bounding when the gain is crisis-sized, a pre-registered population replication on Freddie Mac, and a pre-registered second-episode test on the pandemic band.
  \item \textbf{Two admissible extensions of the instrument, with proofs} (Sections~\ref{sec:vector}, \ref{sec:conditioning}): vector-valued exposure inputs with a coordinatewise monotonicity guarantee, and exogenous conditioning of decay gates, severities, thresholds, and peak memory that is shape-free in context without weakening the guarantee. Both extensions are exercised by the real-data studies; neither is used by the companion synthetic suite.
  \item \textbf{An evaluation template for governed pipelines}: out-of-fold learned features, matched-admissibility baselines including a supervised hand-aggregate champion and a $420$-column hand-interaction control, paired-bootstrap inference, and internal pre-registration with a public audit trail (Appendix~\ref{app:prereg-audit}).
\end{itemize}

\paragraph{Paper organization.}
Section~\ref{sec:related-work} reviews related work, including the macro-conditioned credit-risk literature this paper's second finding extends. Section~\ref{sec:method} recaps the base architecture and presents the two extensions with their guarantees. Section~\ref{sec:real} reports the empirical studies on American Express, Home Credit, Fannie Mae, Freddie Mac, and Home Credit Stability data. Section~\ref{sec:discussion} discusses magnitudes, economic translation, and limitations, and Section~\ref{sec:conclusion} concludes. Full proofs are in Appendices \ref{app:bases} and \ref{app:cond-proofs} and the pre-registration audit trail in Appendix~\ref{app:prereg-audit}; supporting tables (the strictness ladder in numbers, the exploratory UCI probe, the Fannie Mae amendment chronicle) and reproducibility details are in the Online Supplement.

\section{Related Work}
\label{sec:related-work}

\subsection{Credit scoring, model governance, and trended data}
\label{sec:rw-governance}

Credit scoring has a mature methodological literature spanning scorecard construction, performance assessment, and portfolio-level modeling \citep{thomas2017credit_scoring,crook2007consumer_credit}, in which monotone relationships between risk drivers and scores are an entrenched expectation: a scorecard in which higher utilization lowers estimated risk would be rejected by validation teams regardless of its fit statistics. Temporal information enters these pipelines as hand-crafted attributes. Bureau trended-data products synthesize 24-month histories of balances, limits, and payments into windowed aggregates---average utilization over set timeframes, counts of months over limit, changes in balances---consumed as static columns \citep{experian_trended3d,vantagescore4_userguide}. Notably, many trended attributes (utilization slopes, months since last delinquency, balance changes) are \emph{not} monotone in exposures and are therefore inadmissible under the governance constraint of Section~\ref{sec:setup}; under that constraint the engineered library shrinks to its windowed-aggregate core, which is precisely the kind of grid the companion study uses as its engineered baseline \citep{instrument2026}. Supervisory guidance reinforces this expectation: the model risk management frameworks of SR 11-7 and OCC 2011-12 \citep{fed2011_sr11_7,occ2011_12} elevate monotonicity from a modeling preference to a governance requirement that must be demonstrated, documented, and preserved across model updates. Our setting follows the common pipeline structure of this domain: the exposure score is an upstream feature feeding downstream probability-of-default models, so its guarantees must be unconditional---a violation introduced upstream cannot be repaired by a downstream model.

\subsection{Monotone model classes}
\label{sec:rw-monotone}

Architectural enforcement of monotonicity in neural networks has a long history, from non-negative-weight feed-forward networks \citep{archer1993monotonic} and min--max networks \citep{sill1997monotonic} through lattice methods \citep{gupta2016lattice,you2017deep} to modern constructions that expand expressiveness while retaining the guarantee \citep{wehenkel2019umnn,liu2020certified_monotonic,runje2023cmnn}. In credit specifically, monotonic neural additive models target regulated scoring directly \citep{chen2022mnam}, and monotonicity has been framed as a first-class model-risk-management property with dedicated evaluation methodology \citep{chen2023monotonicity_mrm}; both operate on static feature vectors rather than exposure sequences. Industrial practice has extended the weight-constraint mechanism to recurrences: weight-constrained monotonic recurrent networks for credit-risk indicators are patented and deployed \citep{boardman2024mrnn}---direct evidence that this domain demands the guarantee at the sequence level---and the companion study evaluates that construction head-to-head against the instrument used here \citep{instrument2026}. In production tabular learning, gradient-boosted trees \citep{chen2016xgboost,ke2017lightgbm} support per-feature monotone constraints enforced during tree construction \citep{lightgbm_monotone}, which has made monotone boosting the de facto standard in credit modeling pipelines. Unconstrained deep sequence models have also been applied to credit directly: convolutional networks over transaction histories for mortgage default \citep{kvamme2018mortgage}, recurrent networks over card-transaction sequences for loan applications \citep{babaev2019etrnn}, and a systematic benchmark finding deep architectures no better than gradient boosting on standard tabular credit data \citep{gunnarsson2021deep}. None of these carries a monotonicity guarantee, so under the constraint of Section~\ref{sec:setup} they are inadmissible as score components; their free-model verdict is nonetheless congruent with this paper's free-frame null---against a mature engineered panel, learned sequence features add little---which makes the governed-frame question, where this paper locates the value, the one that literature has not asked. The recurrent instrument used here comes from the companion study \citep{instrument2026}: its hidden state consists of interpretable classical risk statistics (a discounted level and a running extreme, in the lineage of exponential smoothing \citep{hyndman2002ets}, RiskMetrics decay \citep{riskmetrics1996}, and the contractual high-water mark \citep{goetzmann2003hwm}), its guarantee holds per input and per MC-dropout sample \citep{gal2016dropout,daw2019pga_uq}, and its predictive cost, interpretability, and expressiveness boundary are characterized there on a controlled synthetic suite. This paper recaps the architecture only to the depth the empirical studies require (Section~\ref{sec:method}) and contributes the two extensions the real-data protocol needs: vector-valued inputs and exogenous conditioning.

\subsection{Macroeconomic conditioning in credit risk}
\label{sec:rw-macro}

Macroeconomic variation has long entered credit-risk models, but at the portfolio or hazard level rather than inside feature construction. Wilson's portfolio credit risk framework conditions sectoral default rates on macroeconomic indices through logistic regressions driven by observed macro factors \citep{wilson1998portfolio}. The Vasicek single-factor model loads each obligor's default probability on one systematic factor, and underlies the asymptotic capital formulas of banking regulation \citep{vasicek2002distribution}. In consumer credit, survival models add macroeconomic covariates to the default hazard and improve out-of-time performance \citep{bellotti2009credit}; in mortgage credit specifically, default and prepayment are competing risks driven by the same macroeconomic option values \citep{deng2000mortgage}, which is why Section~\ref{sec:real-fannie} tests the headline gain against the label's prepayment composition. Rating practice distinguishes point-in-time calibration, which tracks current macroeconomic conditions, from through-the-cycle calibration, which averages across them; the choice governs how quickly scores respond to the regime \citep{aguais2004pit}. Two supervisory frameworks have since made macroeconomic conditioning mandatory practice rather than a modeling choice: expected-credit-loss provisioning under IFRS~9 and CECL, whose point-in-time default probabilities load on forward-looking macroeconomic scenarios \citep{krueger2018provisioning}, and supervisory stress testing, in which consumer-credit models with macroeconomic covariates are re-evaluated under prescribed adverse macro paths \citep{bellotti2013stress}.

These lines share a delivery mechanism: the macroeconomic series enters as a covariate, factor, or scenario input of the risk model itself. None of them conditions the feature-level temporal aggregation of an exposure history; none carries a per-input monotonicity guarantee through the conditioning; and none runs the control this paper runs---the same series delivered to the downstream model as free columns versus the same series conditioning a guaranteed recurrence. Section~\ref{sec:real-fannie} shows the two deliveries have opposite signs under crisis shift, which is invisible to any design that only ever delivers macro information one way.

\paragraph{Positioning.}
The literature provides monotone architectures for static inputs, monotone boosting over hand-crafted temporal aggregates, and macro-conditioned portfolio and hazard models. To our knowledge, no prior work measures the value of \emph{learned} monotone temporal aggregation as a function of governance-frame strictness, and no prior work conditions a monotonicity-guaranteed recurrence on exogenous macroeconomic series with a delivery control. Those two measurements, with the admissibility proofs that make the second one well-posed, are the contribution of this paper.

\section{Monotone Exposure Scoring: Architecture and Extensions}
\label{sec:method}

This section states the base architecture and its guarantee compactly---the full development, including the admissible basis library, expressiveness analysis, and synthetic validation, is in the companion study \citep{instrument2026}---and then presents the two extensions this paper contributes and its empirical studies exercise: vector-valued exposure inputs (Section~\ref{sec:vector}) and exogenous macroeconomic conditioning (Section~\ref{sec:conditioning}).

\subsection{Problem setup and monotonicity constraint}
\label{sec:setup}

We observe a non-negative exposure time series $x_{1:T} = (x_1, \dots, x_T)$ with $x_t \ge 0$ and compute a scalar exposure score $s=f(x_{1:T})$ intended for use as an upstream feature in a downstream risk model.

Monotonicity is enforced as a structural constraint: for the coordinatewise partial order
\begin{equation}
  x_{1:T} \preceq x'_{1:T} \iff x_t \le x'_t \; \forall t,
\end{equation}
we require
\begin{equation}
  x_{1:T} \preceq x'_{1:T} \Rightarrow f(x_{1:T}) \le f(x'_{1:T}).
\end{equation}

\paragraph{Governance admissibility.}
We treat the monotonicity requirement as a normative commitment rather than an empirical hypothesis: a component may enter the score only if it is provably monotone non-decreasing in every exposure input. We call such components \emph{admissible}. Statistics that are informative but not monotone (e.g., variance, trend, or time since last breach) may be used for diagnostics, but are excluded from the score itself. This separation lets the model remain auditable---each score component carries a one-line monotonicity proof---without discarding non-monotone information from the analysis.

\subsection{Base architecture and guarantee}
\label{sec:architecture}

The base model uses a low-dimensional, economically interpretable recurrent state $h_t = [e_t, m_t]^\top$, where $e_t$ is a recency-weighted exposure (EWMA-like) state and $m_t$ is a high-water mark (HWM) state:
\begin{align}
  e_t &= \alpha e_{t-1} + g_1(x_t), \qquad 0<\alpha<1, \\
  m_t &= \max\big(m_{t-1},\; g_2(x_t)\big).
\end{align}
The final score is computed as $s = F(e_T,m_T)$. The two states have direct economic readings: $e_t$ captures the level, persistence, and recency of exposure through exponential decay, while $m_t$ records the worst exposure observed so far, so that severe events permanently affect the score.

\paragraph{Admissible transforms and head.}
The transforms $g_1,g_2:\mathbb{R}_{\ge 0}\to\mathbb{R}$ are monotone non-decreasing (and typically non-negative). Admissible parametric forms include the identity, the logarithm $\log(1+x)$, the hinge $(x-\tau)^+$ with $\tau\ge 0$, and the power $x^p$ with $p \ge 1$. A learned alternative is a \emph{monotone MLP}: a feed-forward network whose weights are constrained non-negative through a parameterization such as $W=\mathrm{softplus}(\Theta)$ and whose activations are monotone non-decreasing; biases are unconstrained. We use tanh activations: with non-negative weights, convex activations confine the network to convex shapes, whereas non-convex monotone activations remove the restriction without weakening the guarantee \citep{runje2023cmnn,instrument2026}. The head $F$ is monotone non-decreasing in each argument (a non-negative linear map or a monotone MLP). The decay $\alpha$ may be fixed in $(0,1)$ or learned via $\alpha=\sigma(\alpha_{\mathrm{raw}})$.

\begin{lemma}[EWMA state is monotone]
\label{lem:ewma}
Assume $g_1$ is monotone non-decreasing and $0<\alpha<1$. Then for all $t$,
$x_{1:t} \preceq x'_{1:t} \Rightarrow e_t(x_{1:t}) \le e_t(x'_{1:t})$.
\end{lemma}
\begin{proof}
By induction on $t$: if $e_{t-1}(x)\le e_{t-1}(x')$ and $x_t\le x'_t$, then
$e_t(x)=\alpha e_{t-1}(x)+g_1(x_t)\le \alpha e_{t-1}(x')+g_1(x'_t)=e_t(x')$.
\end{proof}

\begin{lemma}[HWM state is monotone]
\label{lem:hwm}
Assume $g_2$ is monotone non-decreasing. Then for all $t$,
$x_{1:t} \preceq x'_{1:t} \Rightarrow m_t(x_{1:t}) \le m_t(x'_{1:t})$.
\end{lemma}
\begin{proof}
By induction on $t$, using that $\max(\cdot,\cdot)$ is non-decreasing in each argument.
\end{proof}

\begin{theorem}[Full score is monotone]
\label{thm:score}
Assume $g_1$ and $g_2$ are monotone non-decreasing and $F$ is monotone non-decreasing in each argument. Then
$x_{1:T} \preceq x'_{1:T} \Rightarrow f(x_{1:T}) \le f(x'_{1:T})$, where $f(x_{1:T}) := F(e_T(x_{1:T}), m_T(x_{1:T}))$.
\end{theorem}
\begin{proof}
By Lemmas \ref{lem:ewma}--\ref{lem:hwm}, $(e_T(x),m_T(x))\preceq(e_T(x'),m_T(x'))$. Since $F$ is monotone in both inputs, applying $F$ preserves order.
\end{proof}

When $g_1$, $g_2$, or $F$ is a monotone MLP, the assumptions of Theorem~\ref{thm:score} are met because an affine map with elementwise non-negative weights preserves the coordinatewise order, elementwise monotone activations preserve it as well, and compositions of order-preserving maps are order-preserving. The guarantee also holds per MC-dropout sample, since dropout only rescales non-negative weights by non-negative constants \citep{gal2016dropout,instrument2026}.

\paragraph{Composition with engineered features.}
The same composition argument extends beyond this architecture: any monotone non-decreasing consumer of any fixed monotone featurization---windowed means and maxima, breach counts, or a gradient-boosted tree ensemble with per-feature non-decreasing constraints \citep{lightgbm_monotone}---yields an end-to-end monotone score. The by-construction guarantee is therefore \emph{not} unique to the learned channel; it is equally available to the deployed pipeline of hand-crafted aggregates feeding a constrained scorecard model. What the architecture adds over that pipeline is (i) \emph{learned} aggregation: the decay $\alpha$, the transforms $g_1, g_2$, and their composition with $F$ are fitted jointly rather than fixed ex ante, and (ii) a closed-form readable score: a trained instance collapses to a monotone function of one EWMA and one HWM with plottable one-dimensional transforms, rather than an ensemble over a multi-column feature library. The companion study quantifies what (i) is worth under controlled synthetic mechanisms \citep{instrument2026}; Section~\ref{sec:real} measures it on production data.

\paragraph{Stability by construction.}
The EWMA component is bounded for $\alpha\in(0,1)$ (for non-negative $g_1$ and $e_0=0$, the geometric sum gives $e_T \le \max_t g_1(x_t)/(1-\alpha)$), while the HWM component is bounded by the maximum transformed exposure. This mitigates numerical instability that can occur in naive monotone RNNs without explicit decay.

\subsection{Vector-valued exposure inputs}
\label{sec:vector}

The base model consumes one scalar exposure per step, but real portfolios carry several exposure channels at once, and the governance constraint extends coordinatewise: for a multivariate exposure path with $x_t \in \mathbb{R}^{C}_{\ge 0}$, the score must be monotone non-decreasing in \emph{every} coordinate $x_{t,i}$, under the order $x_{1:T} \preceq x'_{1:T} \iff x_{t,i} \le x'_{t,i}$ for all $t, i$. The recurrences of Section~\ref{sec:architecture} are unchanged; only the arity of the increment transforms generalizes.

\paragraph{Admissible multivariate increments.}
The admissible increment class under vector input is the non-negative-weight monotone MLP $g:\mathbb{R}^{C}_{\ge 0}\to\mathbb{R}^{J}_{\ge 0}$---non-negative weight matrices, monotone non-decreasing activations, and a final softplus for non-negativity---so that each output coordinate $g_j$ is an isotone map of the full input vector and drives one channel's scalar recurrence exactly as $g(x_t)$ does in the scalar model. The resulting states are temporal aggregates of \emph{learned cross-channel interactions}, which is what distinguishes the vector model from a bank of per-channel scalar models. The parametric scalar transforms (identity, log, hinge, power) are not admissible increments here: applied coordinatewise they produce a vector where the recurrence needs a scalar.

\paragraph{Permitted bases.}
Under vector input the state may consist only of EWMA, HWM, and \emph{low-water-mark} (LWM) channels, the last defined by $\ell_t = \min\big(\ell_{t-1},\, g_j(x_t)\big)$ and monotone by the same argument as the HWM with $\min$ in place of $\max$; the LWM is equally admissible in the scalar model (Lemma~\ref{lem:lwm}, Appendix~\ref{app:bases}). The remaining scalar bases of the admissible library (cumulative sums, breach counts, cumulative excesses, run lengths; see \citealp{instrument2026}) consume the raw step value or threshold comparisons on it, which are not defined on $\mathbb{R}^{C}$; they are excluded.

\paragraph{Sign orientation.}
Coordinatewise monotonicity presumes every input channel is oriented so that larger means riskier. A channel whose empirical risk direction is decreasing enters through the fixed order-reversing map $x' = \max(0,\, c - x)$, where $c$ is a constant (an upper quantile of the observed values) frozen on development data and applied as preprocessing outside the model. The model is then monotone non-decreasing in $x'$ and hence non-increasing in the raw channel---the prespecified economic direction; because $c$ is a dataset statistic rather than a per-example quantity, no dependence on the protected path is introduced.

\begin{proposition}[Vector-input monotonicity]
\label{prop:vector}
Let $g^{(1)}, g^{(2)}:\mathbb{R}^{C}_{\ge 0}\to\mathbb{R}^{J}_{\ge 0}$ be non-negative-weight monotone MLPs, let the state consist of EWMA, HWM, and/or LWM channels each driven by one output coordinate of $g^{(1)}$ or $g^{(2)}$, and let the head $F$ be monotone non-decreasing in every channel. Then the score is monotone non-decreasing in every coordinate $x_{t,i}$ of the exposure path.
\end{proposition}
\begin{proof}
See Appendix~\ref{app:bases}.
\end{proof}

\subsection{Exogenous macroeconomic conditioning}
\label{sec:conditioning}

The base recurrence applies one fixed decay rate to the entire history and one fixed increment scale to every step. Under a macroeconomic regime shift both choices are wrong in a specific way: how long exposure history stays informative, and how much a given exposure counts when it occurs, are properties of the regime, not of the borrower. This subsection conditions the recurrence on an \emph{exogenous context} $z_t$---calendar time, position, or published macroeconomic series---that is never a function of the protected exposure path $x_{1:T}$. We prove that the per-input guarantee survives, for \emph{any shape} of the conditioning, and we show that exogeneity is the load-bearing assumption rather than a convenience. Main-text proofs are one-line sketches; full proofs are in Appendix~\ref{app:cond-proofs}.

\paragraph{Gated decay.}
A \emph{gated EWMA} replaces the constant decay with a per-step gate:
\begin{equation}
\label{eq:gated-ewma}
  e_t = a_t\, e_{t-1} + g_1(x_t), \qquad a_t = a(x_t,\ t,\ z_t) \in (0,1), \qquad e_0 \ge 0.
\end{equation}
The gate is admissible under two conditions: \textbf{(D1)} range, $0 < a_t < 1$ (e.g., $a_t = \sigma(\cdot)$), which keeps the channel contractive and bounded; and \textbf{(D2)} protected-input monotonicity, $a_t$ non-decreasing in $x_t$ and seeing $x_t$ only at its own step, while its dependence on the exogenous arguments $(t, z_t)$ is unconstrained in shape. Economically, the gate spans three mechanisms in one construction: exposure-gated persistence (``running up balances makes history stickier''), position-dependent decay, and macro-conditioned decay (``in a downturn, old exposure stays informative longer''). The plain EWMA is the constant-gate member $a_t \equiv \alpha$.

\begin{lemma}[Gated EWMA is monotone in the exposure path]
\label{lem:gated-ewma}
Under (D1)--(D2), with $g_1$ monotone non-decreasing and non-negative, the terminal state $e_T$ of Equation~\eqref{eq:gated-ewma} is non-decreasing in every $x_k$.
\end{lemma}
\begin{proof}[Proof sketch]
Unrolling gives $e_T = \big(\prod_{r=1}^T a_r\big) e_0 + \sum_{s=1}^T \big(\prod_{r=s+1}^T a_r\big) g_1(x_s)$; the input $x_k$ appears only in the gate $a_k$ and the increment $g_1(x_k)$, each carrying a non-negative coefficient, and collecting terms yields $\partial e_T/\partial x_k = C_k\big[\, e_{k-1}\,\partial a_k/\partial x_k + g_1'(x_k)\,\big] \ge 0$ with $C_k = \prod_{r=k+1}^T a_r \ge 0$. Full proof in Appendix~\ref{app:cond-proofs}.
\end{proof}

Two remarks calibrate the conditions. First, \emph{exogenous gating is shape-free}: if the gate depends only on $(t, z_t)$, then $\partial a_k/\partial x_k = 0$ and monotonicity holds for any shape of $a(t, z_t)$---this is what licenses macro-conditioned decay with full functional freedom. Second, \emph{the direction condition (D2) is necessary}: a gate that decreases in $x_t$ (the naive ``forget faster when recent exposure is high'') breaks monotonicity, because raising a late input can shrink the retained weight on a large accumulated past by more than its own increment adds; Appendix~\ref{app:cond-proofs} gives a two-step counterexample.

\paragraph{Leaky peak memory and the gated leaky-max.}
The HWM's ``extremes are never forgotten'' reading can be relaxed to ``extremes fade with recency'' through the \emph{leaky-max}
\begin{equation}
  r_t = \max\big(\beta\, r_{t-1},\; g_2(x_t)\big), \quad 0 < \beta \le 1
  \quad\Longrightarrow\quad
  r_T = \max_t\, \beta^{T-t} g_2(x_t),
\end{equation}
the recency-discounted peak; the plain HWM is the $\beta = 1$ member. Replacing the constant $\beta$ with an exogenous gate $b(z_t) \in (0,1)$ gives the \emph{gated leaky-max}: in a downturn $b \to 1$ and peaks linger, in recovery they fade fast.

\begin{lemma}[Gated leaky-max is monotone]
\label{lem:gated-leakymax}
Let $r_t = \max\big(b(z_t)\, r_{t-1},\; g_2(x_t)\big)$ with $b(z_t) \in (0,1]$ exogenous (the pinned HWM is the $b \equiv 1$ member), $r_0 = 0$, and $g_2$ monotone non-decreasing and non-negative. Then $r_T = \max_{1 \le t \le T} \big(\prod_{s=t+1}^{T} b(z_s)\big)\, g_2(x_t)$, and $r_T$ is non-decreasing in every $x_k$, for any shape of $b$.
\end{lemma}
\begin{proof}[Proof sketch]
Each closed-form term is a non-negative multiple of a monotone non-decreasing $g_2(x_t)$, and a max of non-decreasing functions is non-decreasing; by exogeneity, raising $x_k$ moves no gate factor. Full proof in Appendix~\ref{app:cond-proofs}.
\end{proof}

\paragraph{The pinned HWM channel.}
At the permanence boundary the two recurrences behave oppositely, and this asymmetry dictates the bank structure used in Section~\ref{sec:real-fannie}. The leaky-max is \emph{bounded} at $\beta = 1$: $r_T = \max_t g_2(x_t)$, independent of $T$---an admissible, qualitatively distinct member that a sigmoid-learned $\beta < 1$ never attains. One channel is therefore \emph{pinned} at $\beta = 1$ (fixed, non-trainable), so the never-forgotten peak is guaranteed by construction rather than learned. The EWMA admits no such anchor: at $\alpha = 1$ it becomes the unbounded cumulative sum and violates the contraction the channel exists to provide, so the EWMA bank consists solely of learned rates interior to $(0,1)$, with no pin.

\paragraph{Economy-conditioned severity.}
The decay gate conditions how long a channel remembers; \emph{severity} conditions how much a step counts when it happens. Replace a channel's increment $g(x_t)$ with $c(z_t)\, g(x_t)$, where $c:\mathcal{Z}\to\mathbb{R}_{\ge 0}$ is a non-negative multiplier driven by the same exogenous context. Economically, a fixed dollar exposure is more severe in a downturn: $c(z_{\text{recession}}) \gg c(z_{\text{boom}})$ makes a \$1k spike in a bad month count like a \$10k spike in a good one, and a max channel freezes the occurrence-time severity into the stored peak.

\begin{lemma}[Exogenous severity is monotone]
\label{lem:severity}
With $g$ monotone non-decreasing and non-negative and $c(z_t) \ge 0$ exogenous, substituting the increment $u_t = c(z_t)\, g(x_t)$ into the (gated) EWMA or the (gated) leaky-max leaves the terminal state non-decreasing in every $x_k$, for any shape of $c$.
\end{lemma}
\begin{proof}[Proof sketch]
$c(z_k)$ is constant in $x_k$, so $x_k$ appears exactly where $g(x_k)$ did, scaled by a non-negative constant; non-negativity of $c$ is necessary because it keeps the increment, hence the accumulant, non-negative. Full proof in Appendix~\ref{app:cond-proofs}.
\end{proof}

\paragraph{Economy-conditioned threshold.}
The hinge increment may take an exogenous, regime-dependent threshold in place of a fixed constant: $g(x_t) = \max\big(0,\, x_t - \tau(z_t)\big)$ with $\tau:\mathcal{Z}\to\mathbb{R}_{\ge 0}$. This is a re-baselining lever, complementary to severity's re-scaling: in a downturn $\tau$ drops, so a moderate spike that was normal in good times now clears the bar and registers positive excess. Severity multiplies all exposure proportionally; threshold shifts the activation point and so bites small exposures disproportionately. The two compose.

\begin{lemma}[Exogenous threshold is monotone]
\label{lem:threshold}
$\max(0,\, x - \tau)$ has slope in $\{0, 1\}$, hence is non-decreasing in $x$ and non-negative for any $\tau$; with $\tau(z_t)$ exogenous, substituting this increment into either channel preserves monotonicity in every $x_k$, for any shape of $\tau$.
\end{lemma}
\begin{proof}[Proof sketch]
Since $\tau(z_t)$ does not depend on $x_t$, the per-step derivative in $x_t$ is unaffected by its value; the constraint $\tau \ge 0$ is a semantic convention (keep the threshold in the non-negative exposure range), not a monotonicity requirement.
\end{proof}

\begin{theorem}[The conditioned bank is monotone]
\label{thm:cond-bank}
Let the model carry a bank of channels combining gated or constant-decay EWMA channels (Lemma~\ref{lem:gated-ewma}) and gated or constant leaky-max channels (Lemma~\ref{lem:gated-leakymax}), with increments optionally severity-scaled (Lemma~\ref{lem:severity}) and/or threshold-shaped (Lemma~\ref{lem:threshold}), all under the stated conditions. If the score head $F$ is monotone non-decreasing in all channels, then $s = F(\dots)$ is monotone non-decreasing in the entire exposure path.
\end{theorem}
\begin{proof}[Proof sketch]
Each channel is monotone in every input by the four lemmas; monotone composition through $F$ as in Theorem~\ref{thm:score}.
\end{proof}

\paragraph{Exogeneity is essential.}
If $c$, $\tau$, or $b$ depended on a protected input, the gated-decay failure family reappears: a severity \emph{decreasing} in $x_t$, a threshold \emph{increasing} in $x_t$, or a leaky-max gate \emph{decreasing} in $x_t$ can each make the relevant derivative negative. The shape-freedom of Lemmas~\ref{lem:severity}--\ref{lem:threshold} and~\ref{lem:gated-leakymax} holds precisely because $z_t$ is exogenous and the score is never differentiated through it. This is the answer to the natural governance question---``how can arbitrary macro dependence be safe inside a guaranteed score?''---the guarantee protects the exposure inputs, and the macro series is not one of them.

\paragraph{Boundary with inadmissible attention.}
Gating is \emph{not} data-dependent attention, and the difference is exactly why one is admissible and the other is not. A gate multiplies a single channel's own non-negative accumulant $e_{t-1}$; it never forms a convex combination ($\sum_j w_j = 1$) across channels, so it has no sign-indefinite reshuffle term. Attention's derivative contribution $\sum_j (\partial w_j/\partial x_t)\, S_j$ is sign-indefinite precisely because $\sum_j \partial w_j/\partial x_t = 0$: raising $x_t$ moves a fixed weight budget from one channel onto another. The gate's per-step contribution $e_{t-1}\, (\partial a_t/\partial x_t) \ge 0$ carries no such sum-to-zero constraint---it raises or holds the persistence of a channel's own past without reallocating anything across channels.

\paragraph{Decay over missing periods.}
With period-anchored sampling, a missing interior period is an unobserved step. Such a step may either freeze the state ($e_t = e_{t-1}$) or decay it for the elapsed period ($e_t = a_t\, e_{t-1}$, the gated update with the increment dropped; $r_t = b(z_t)\, r_{t-1}$ for the leaky-max). Decay is the calendar-faithful choice---it lets recency fade across a gap rather than treating the gap as if no time passed---and it preserves monotonicity: the gap update introduces no dependence on the protected exposure, and dropping a non-negative increment only removes a non-negative term. A pinned $\beta = 1$ channel is unchanged by the gap, preserving the never-forgotten peak. This choice matters empirically: the short-history population is where the conditioned bank posts its largest gains (Section~\ref{sec:real-fannie}).

\paragraph{The deployed bank.}
The Fannie Mae and Freddie Mac studies (Section~\ref{sec:real-fannie}) instantiate Theorem~\ref{thm:cond-bank} as a per-channel \emph{macro-conditioned multihead bank}. For each oriented exposure channel, the bank comprises a severity sub-bank ($c(z)$ scaling the increments), a decay-gate sub-bank ($a(z)$ replacing the constant $\alpha$), a threshold sub-bank ($\tau(z)$ hinge increments), a macro-gated leaky-max ($b(z)$), and the pinned $\beta = 1$ HWM channel, combined by a monotone head into one per-channel score. In the released configuration, each of the three conditioned EWMA sub-banks carries four heads, alongside a plain four-head static-decay EWMA sub-bank, four macro-gated leaky-max channels, and the pinned HWM---twenty-one states per channel, combined by a hierarchical monotone head ($[16,8]$ hidden units, tanh) over monotone-MLP increments ($[16,8]$, tanh) on log-transformed exposure. Each conditioner is a zero-initialized affine map of the exogenous context alone: $c(z) = \operatorname{softplus}(c_0 + w_c^\top z)$, $a(z) = \sigma(a_0 + w_a^\top z)$ anchored on the learned static decay logits, $\tau(z) = \operatorname{softplus}(\tau_0 + w_\tau^\top z)$, and $b(z) = \sigma(b_0 + w_b^\top z)$. The zero initialization makes a fresh conditioned bank coincide exactly with its static anchor ($c = 1$, $\tau \approx 0$, decays at their learned constants), so the conditioning departs from the static family only as training learns it. The context $z$ holds two published series per loan-month (the national unemployment rate and a lagged home-price-index log-change; state-level variants in the geography test). The \emph{static bank} used as a control is the same family with all conditioning removed---four static-decay EWMA heads, four learned leaky-max heads, and the pinned HWM, sharing increments, head, input transform, and gap handling, with no $z$; at any fixed $z$ each conditioned sub-bank collapses to a member of this static class, so the conditioning contrast measures the bank's response to the regime rather than the presence of a bank.

\section{Empirical Studies}
\label{sec:real}

This section reports what the guarantee and the learned aggregation are worth on production-scale credit data. Sections~\ref{sec:real-uplift}--\ref{sec:real-hc} evaluate the vector-input extension of Section~\ref{sec:vector} on two public credit datasets (plus an exploratory single-channel probe on a third; Online Supplement \S S2), with leakage-safe out-of-fold features and paired-bootstrap inference throughout. Two pre-registered extensions then probe the weakest axes of what those studies find on independent datasets: a severe, genuine vintage shift on Fannie Mae loan performance (Section~\ref{sec:real-fannie}), where the conditioning extension of Section~\ref{sec:conditioning} enters, is replicated on Freddie Mac, and is scope-tested on the pandemic band; and a replication of the strictness ladder at thirteen-channel breadth on the 2024 Home Credit Stability competition, scored by its official temporal-stability metric (Section~\ref{sec:real-hcs}). Section~\ref{sec:real-synthesis} assembles the verdicts.

\subsection{Setup}
\label{sec:real-setup}

\paragraph{Datasets, regimes, and metric.}
The primary dataset is the American Express default-prediction panel: $458{,}913$ customers, up to $13$ monthly statements each, a standard $188$-column engineered feature panel, and the official competition metric $M = \tfrac{1}{2}(\text{normalized weighted Gini} + \text{top-}4\%\text{ default capture})$ with negatives downweighted $20{\times}$. We evaluate on three regimes: \textsc{iid} (random held-out customers), \textsc{time} (a short-history cohort, customers with fewer than the full $13$ statements---a deployment-onboarding shift), and \textsc{peak} (the top $30\%$ of a composite high-water severity index---a stress-cohort shift). The replication dataset is Home Credit: a credit-card panel ($86{,}905$ clients, $8$ monthly balance-activity channels) and an installments-payments panel ($291{,}643$ clients, $2$ raw timing channels: days late and payment shortfall), evaluated by AUC on random and shortest-history-cohort splits.

\paragraph{The learned features.}
Fifteen AMEX delinquency- and risk-category channels are first \emph{oriented} by the dev-frozen antitone map of Section~\ref{sec:vector} (two channels flip on the \textsc{iid} and \textsc{time} development pools, four on \textsc{peak}; zero of eight on Home Credit). Orientation is a development-pool statistic, not a runtime choice: each regime's development pool fixes the map ex ante, exactly as a production system would freeze it on its own development data (\textsc{peak} under the \textsc{iid} orientation is not run). A single vector-input model consumes all fifteen oriented trajectories through monotone-MLP increments ($[32,16]$, tanh) and maintains one EWMA, one HWM, and one LWM state; the three terminal states enter downstream models as three features (``joint''). A separate bank of per-channel scalar models supplies $45$ per-channel states (three per channel: the terminal EWMA, HWM, and LWM states of a scalar model fit to that channel alone) where noted. All recurrent features are cross-fit out-of-fold: each customer's features come from a model whose training folds excluded that customer, and test-cohort features from a dev-pool refit. Downstream consumers are monotone-constrained XGBoost models (all features signed $+1$), so every pipeline built on the learned states and level aggregates is monotone end-to-end by Proposition~\ref{prop:vector} and the composition argument of Theorem~\ref{thm:score}; pipelines containing the supervised summary scalars inherit one disclosed exception, the sign-free volatility aggregate inside each scalar (see \emph{Baselines} below). Per-coordinate bump tests on every deployed vector-input model (nine AMEX refits across regimes and RNN seeds, two Home Credit joints) record a violation rate of exactly zero, consistent with Proposition~\ref{prop:vector}.

\paragraph{Baselines.}
The incumbent to beat is \textsc{sup15}: for each of the fifteen channels, an out-of-fold monotone-XGBoost scalar trained on that channel's five hand aggregates (last, mean, min, max, std)---supervised compression of hand-crafted aggregates, the strongest summary baseline we found and the real-data analog of the engineered feature grids studied synthetically in the companion paper \citep{instrument2026}. One convention inside these scalars is disclosed rather than hidden: the four level aggregates are constrained $+1$, but the volatility aggregate is deliberately left sign-free, because volatility is informative yet not monotone in the exposure path. Strictly, therefore, a pipeline containing a summary scalar is not provably monotone end-to-end---an upstream violation cannot be repaired downstream (Section~\ref{sec:rw-governance})---and the same convention builds the mortgage and stability summaries of Sections~\ref{sec:real-fannie} and~\ref{sec:real-hcs}. The convention strengthens the baseline, and a sensitivity run over the same confirmatory seed grid as the primary contrast prices it: the volatility freedom is worth $+0.0093$ to $+0.0100$ $M$ to the champion on \textsc{iid} and $+0.0109$ to $+0.0114$ on \textsc{peak}, CI-clear at every one of the six seed points, and against the fully admissible four-aggregate champion (last, mean, min, max, all $+1$) the joint-feature uplift widens to $+0.0088$ to $+0.0110$ on \textsc{iid} and $+0.0061$ to $+0.0116$ on \textsc{peak}, likewise CI-clear at every seed point; on \textsc{time} both contrasts straddle zero at every seed point, consistent with that column's instability. Quoting the primary contrast against the stronger sign-free champion throughout is thus the conservative choice. A second convention is priced rather than assumed: the champion \emph{compresses} each channel's aggregates into one scalar, and a compression control in Section~\ref{sec:real-strictness} hands the downstream the same aggregates uncompressed---a stronger frame whose own contrast bounds what compression could account for. A hand-interaction control (\textsc{product aggregates}) applies the same window aggregates to all $105$ pairwise product series of the oriented channels ($420$ columns), testing whether hand-built cross-channel interactions replicate the learned ones. The saturation reference is the full $188$-column engineered panel, consumed unconstrained.

\paragraph{Protocol discipline.}
The primary contrast---\textsc{sup15}${}+{}$joint versus \textsc{sup15}, $18$ versus $15$ features, both fully signed---was identified on an initial seed and then confirmed on independent seeds along both axes: the confirmatory grid of Table~\ref{tab:real-amex-uplift} (three RNN seeds at the anchor XGBoost seed, and three further XGBoost seeds at the anchor RNN seed) was fixed before any of its fits ran. No analysis in this paper was externally registered; \emph{prespecified} and \emph{pre-registered} throughout mean fixed internally---committed to a version-controlled design file---before the corresponding fits ran. Contrasts marked exploratory below were selected after seeing results. One conservative asymmetry is documented rather than repaired: the out-of-fold scalars are cross-fit on \textsc{iid} folds regardless of evaluation regime, which under the shift regimes mildly flatters the \emph{baseline} (its aggregates are cheaper to re-fit than the recurrent features), so the reported uplifts are, if anything, understated.

\paragraph{The inference standard.}
Throughout the paper, a contrast is \emph{CI-clear} when its paired-bootstrap $95\%$ confidence interval excludes zero, and the primary decision criterion is that a contrast is CI-clear at every seed (or seed pair) in a prespecified grid. This criterion is a \emph{seed-robustness standard}, not a familywise-error-controlled joint test. Each interval controls its own coverage; requiring all of them to clear zero guards against seed-specific artifacts, but it does not bound the joint type-I error of the family of comparisons, and the implied joint test is conservative in some directions and liberal in others. For the headline mortgage contrasts, we complement the per-seed standard with the formal alternative: a normal--normal random-effects model over the five-RNN-seed grid (DerSimonian--Laird $\hat\tau^2$ with within-seed standard errors recovered from the bootstrap intervals, and modified Knapp--Hartung intervals on $t_4$ quantiles, appropriate at five seeds), reported in Table~\ref{tab:seed-meta}. Its pooled interval carries one caveat in the other direction: the per-seed contrasts share test rows, so the shared sampling component does not average away across seeds and the pooled interval is anti-conservative on it; the between-seed standard deviation $\hat\tau$ is the quantity the per-seed standard cannot report, and $\hat\tau = 0$ rows should be read as no \emph{detectable} seed heterogeneity at $k = 5$, not none. Where the text says a result ``clears zero at every seed,'' it should be read as robustness evidence under this standard, and the pre-registration audit trail (Appendix~\ref{app:prereg-audit}) records which contrasts were confirmatory under it.

\subsection{The admissible-frame uplift}
\label{sec:real-uplift}

Table~\ref{tab:real-amex-uplift} reports the primary contrast across three RNN seeds and four XGBoost seeds. Adding the three joint features to the fifteen-summary champion is worth ${+}0.0043$--$0.0051\,M$ on \textsc{iid} and ${+}0.0026$--$0.0047\,M$ on \textsc{peak}, with every $95\%$ confidence interval excluding zero across both seed axes---a robust, confirmatory result. On \textsc{time} the point estimates are positive at every seed, but one RNN reseed collapses the effect to a statistical zero ($+0.0009$), and at higher bootstrap resolution ($B = 10{,}000$; Online Supplement \S S4) a second cell's lower bound crosses zero, so we report the \textsc{time} column as marginal and unstable under reseeding, and claim nothing for it.

Two controls locate what the three features contain. First, the hand-interaction control fails to replicate them: the $420$ product-aggregate columns never significantly beat \textsc{sup15} (Table~\ref{tab:real-amex-controls}, third row), while the learned block does---so the uplift is not ``interactions, available by hand''; it is learned temporal aggregation of cross-channel structure that the windowed products miss. Second, in exploratory contrasts (selected post hoc; Table~\ref{tab:real-amex-controls}, first two rows) the wider $48$-feature block beats \textsc{sup15} on \textsc{iid} and \textsc{peak} and beats the product-aggregate control outright on \textsc{peak}, but loses badly on \textsc{time}---the per-channel states, not the joint features, are the drag, which is why the confirmatory contrast isolates the three joint features.

\begin{table}[t]
\centering
\caption{AMEX primary contrast: paired $\Delta M$ from adding the three joint cross-channel features to the fifteen-summary champion (\textsc{sup15}${}+{}$joint $-$ \textsc{sup15}, $18$ vs.\ $15$ features, both fully signed), across three RNN seeds (at XGB seed $0$) and three further XGBoost seeds (at RNN seed $0$), with paired-bootstrap 95\% CIs. \textsc{time} is reported, not claimed (Section~\ref{sec:real-uplift}).}
\label{tab:real-amex-uplift}
\footnotesize
\setlength{\tabcolsep}{4pt}
\adjustbox{max width=\linewidth}{\input{tables/real_amex_uplift}}
\end{table}

\begin{table}[t]
\centering
\caption{AMEX controls (exploratory; the confirmatory contrast is Table~\ref{tab:real-amex-uplift}). The $48$-feature learned block ($45$ per-channel states plus $3$ joint features) against the supervised summaries and against the $420$-column hand product-aggregate control, and the control against the summaries.}
\label{tab:real-amex-controls}
\footnotesize
\setlength{\tabcolsep}{4pt}
\adjustbox{max width=\linewidth}{\input{tables/real_amex_controls}}
\end{table}

\subsection{No uplift on the unconstrained panel}
\label{sec:real-panel}

We state the negative result plainly before interpreting it. Added to the unconstrained $188$-column panel---the frame a production incumbent would actually use---the recurrent features contribute nothing significant: all $48$ features move $M$ by $+0.0011$ (\textsc{iid}), $+0.0020$ (\textsc{time}), and $-0.0012$ (\textsc{peak}), every interval straddling zero (Table~\ref{tab:real-amex-panel}). On \textsc{time}, the one regime the panel does not saturate, the three joint features alone do clear zero ($+0.0032$)---but so does the \textsc{sup15} block ($+0.0048$), and given the summaries the joint features add exactly nothing ($+0.0000$): on the free panel the learned features are a \emph{substitute} for supervised hand aggregates, not a complement. Saturation is a property of the frame, not a failure specific to the recurrence---even the hand-crafted champion block adds nothing significant to the free panel on \textsc{iid} and \textsc{peak}. The free frame's accuracy is not free of admissibility, however: under a matched per-coordinate bump protocol, the unconstrained panel model moves its score \emph{against} the development-frozen monotone direction on $9.1$--$15.2\%$ of upward perturbations (evaluated on both the $28$-column confident and $87$-column broad sign sets, across all three regimes), where every sign-constrained arm records exactly zero. The saturation the free panel offers is bought by violating roughly a tenth of the monotone prior its own development data supports.

\begin{table}[t]
\centering
\caption{AMEX full-panel honesty check: paired $\Delta M$ from adding blocks to the unconstrained $188$-column panel (Section~\ref{sec:real-panel}).}
\label{tab:real-amex-panel}
\footnotesize
\setlength{\tabcolsep}{4pt}
\adjustbox{max width=\linewidth}{\input{tables/real_amex_panel}}
\end{table}

\subsection{The strictness ladder}
\label{sec:real-strictness}

The two preceding results---a robust uplift at matched admissibility, none on the free panel---are reconciled by a single experiment. We evaluate the same three joint features in five frames of decreasing freedom: the \emph{free} $188$-column panel; \emph{confident} ($28$ columns signed where the empirical monotone direction is unambiguous, the rest free); \emph{broad} ($87$ columns signed); \emph{strict} (only the $87$ signable columns retained, all signed); and \emph{summaries} (the fifteen supervised summaries, all signed). The summaries frame is not a straw construction: it is the sub-score and trended-attribute shape in which temporal information already ships to production scorecards (Section~\ref{sec:rw-governance}), and the strict frame is the panel analog in which every retained column carries a signed commitment.

The result is a ladder, and the ladder has a price attached. Figure~\ref{fig:strictness} shows the result; Supplement Table~S1 reports the numbers. On the two claimed regimes the uplift from the three learned features rises with strictness: statistically zero across the free, confident, and broad frames, boundary in the strict frame (intervals graze zero; we do not claim these cells), and CI-clear in the summaries frame. On \textsc{time} the ladder is flatter: every frame's cell is positive at the anchor seed and even the free-panel cell clears zero (Section~\ref{sec:real-panel}), but the column remains unclaimed for its reseeding instability (Section~\ref{sec:real-uplift}). The pattern holds for the hand-crafted block too: the \textsc{sup15} summaries, worthless on the free panel on \textsc{iid} and \textsc{peak}, re-become additive in the strict frame ($+0.0039$ on \textsc{iid}, interval clear of zero). The lower panel prices the frames themselves: strictness costs base accuracy (the strict frame gives up ${\sim}0.014\,M$ on \textsc{iid} relative to free, the summaries frame ${\sim}0.090\,M$), and it is exactly this surrendered freedom that the learned features partially buy back. The value of learned monotone recurrence is the price of the freedom the model is not allowed to have: where governance binds---every feature signed, or only auditable summaries admitted---the learned features carry signal nothing else in the frame supplies; where the model is free, a large enough engineered panel absorbs it. The ladder is not the mechanical headroom of a weaker baseline, under which any informative addition would clear zero against the weakest frame: the $420$-column hand-interaction control faces the same summaries frame and never clears (Table~\ref{tab:real-amex-controls}), and on the free panel the same three features are subsumed rather than weak (mechanism paragraph below)---what rises with strictness is the residual content of exactly these features.

\paragraph{The compression control: aggregates uncompressed.}
The summaries frame compresses each channel's five aggregates into one supervised scalar, so its uplift could in principle be recovery of information the sub-score bottleneck discards rather than learned temporal content. Handing the downstream the same aggregates uncompressed refutes that reading and inserts a measured rung between summaries and strict. Under the champion's own sign convention (level aggregates signed, volatility free), the $75$-column aggregate frame is a stronger incumbent than the compressed champion (base $M$ $0.7050$ on \textsc{iid}, ${+}0.0096$ to ${+}0.0121$ above \textsc{sup15} and ${+}0.0080$ to ${+}0.0091$ on \textsc{peak}, CI-clear at all six seed points of the confirmatory grid)---and the three joint features add nothing to it on \textsc{iid} (straddling at all six points). Dropping the fifteen sign-free volatility columns---the fully admissible frame, sixty columns all signed---restores the uplift at every seed point: ${+}0.0059$ to ${+}0.0086$ on \textsc{iid} and ${+}0.0047$ to ${+}0.0090$ on \textsc{peak}, roughly one-and-a-half to two times the summaries-frame contrast (CI-clear at all six points on \textsc{iid}; all six for the fully signed frame and five of six for the std-free frame on \textsc{peak}). The uplift is therefore not compression recovery---it survives, and grows, when the aggregates arrive raw at full admissibility---and the aggregate family reproduces the ladder in miniature: freeing fifteen volatility columns absorbs exactly what the three learned features carry. On \textsc{time} the joint features \emph{hurt} the std-free aggregate frame (CI-clear negative at five of six points) and straddle zero on the fully signed one, consistent with that column's disclosed instability. The practitioner readings are two-sided: a regime that tolerates the sign-free volatility convention should ship the uncompressed aggregates and expect nothing further from these features in distribution, while a fully signed pipeline gains more from the three features over the raw aggregates than over the compressed sub-scores.

\paragraph{Mechanism: subsumption, not weakness.}
Why are the learned features substitutes on a free panel and complements in a strict frame? The companion study answers constructively, by direct functional regression \citep{instrument2026}: a single learned monotone channel reproduces the entire geometrically weighted separable family of hand-crafted statistics---EWMAs at arbitrary decays, cumulative sums, breach counts, cumulative excesses---to rank equivalence (Spearman $\rho = 1.000$ on the separable family), with its boundary at consecutivity, time localization, and the exposure floor (the running minimum). A saturated engineered panel therefore already carries most of what the learned channel computes, and the free-frame zero is subsumption rather than absence of signal; a strict frame removes the redundant carriers, and the learned features become the only remaining source of that content.

The ladder's cross-dataset replication---at thirteen channels, under an official temporal-stability metric---is deferred to Section~\ref{sec:real-hcs}, because its temporal half belongs to the regime-shift story of Section~\ref{sec:real-fannie}.

\begin{figure}[t]
\centering
\includegraphics[width=\linewidth]{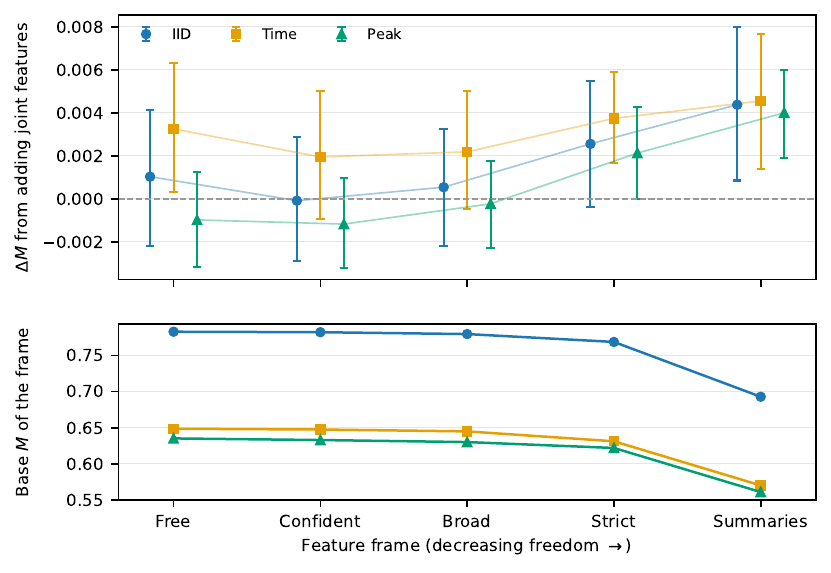}
\caption{The strictness ladder on AMEX. Top: paired $\Delta M$ (with $95\%$ CIs) from adding the three joint cross-channel features to each of five feature frames of decreasing freedom, per regime. Bottom: base $M$ of each frame. On the claimed regimes (\textsc{iid}, \textsc{peak}) the uplift rises as the frame surrenders accuracy to constraint; the unclaimed \textsc{time} column is positive in every frame (Supplement Table~S1).}
\label{fig:strictness}
\end{figure}

\subsection{Home Credit: the breadth boundary}
\label{sec:real-hc}

Home Credit replicates the protocol where its assumptions fail, and the failures are as informative as the AMEX success (Table~\ref{tab:real-hc}). On the $8$-channel card panel, hand aggregates of the raw channels are strongly additive over static covariates ($+0.023$ AUC), but per-channel recurrent features \emph{hurt} under the short-history shift ($-0.0031$ and $-0.0022$, intervals entirely negative)---short histories starve recency states. The joint cross-channel features are the best recurrent rung (significantly above the per-channel states under shift, $+0.0017$) yet only claw back to neutral. On the $2$-channel installments panel the verdict is sharper: hand aggregates of the two raw timing signals are unambiguously real ($+0.0094$ and $+0.0069$ over covariates), while every recurrent rung is a \emph{precisely measured zero}---intervals of half-width ${\sim}0.001$ or less around zero at $291{,}643$ clients. The null is not a power artifact; at $3.4{\times}$ the card sample it is the tightest estimate in the study. Two readings follow. The cross-channel uplift requires channel \emph{breadth}: fifteen interacting delinquency channels carry learnable cross-channel temporal structure; two nearly-redundant timing channels do not. And the hypothesis that AMEX's engineered features had merely ``lost'' recoverable raw-signal information is refuted---given supervised aggregates of the raw signal, learned recurrence adds nothing on these short, narrow panels either.

\begin{table}[t]
\centering
\caption{Home Credit replication (paired $\Delta$AUC with $95\%$ CIs). Card panel ($8$ channels): recurrent rungs versus the hand-aggregate baseline. Installments panel ($2$ raw timing channels): the same ladder at $291{,}643$ clients (Section~\ref{sec:real-hc}).}
\label{tab:real-hc}
\footnotesize
\setlength{\tabcolsep}{4pt}
\adjustbox{max width=\linewidth}{\input{tables/real_hc_ladders}}
\end{table}

An exploratory single-channel probe on the UCI Taiwan card panel (Online Supplement \S S2) bounds the boundary from below: at breadth one the learned-versus-aggregate direction is a property of the channel, not of the model class.

\subsection{Fannie Mae: a severe vintage shift}
\label{sec:real-fannie}

\subsubsection{Design and pre-registration}

The AMEX \textsc{time} regime is a history-length proxy for temporal shift, and its uplift was unstable under reseeding; Fannie Mae single-family loan performance data supplies the real thing. We snapshot every loan at seasoning month twelve and predict 90-plus-days delinquency within the following twenty-four months, from monotone monthly channels plus origination covariates (FICO, LTV, DTI, balance). The label is the standard fixed-horizon ever-90 definition, so a loan that terminates inside the window without reaching 90 days---prepayment included---counts as a non-default; the competing prepayment risk \citep{deng2000mortgage} is absorbed into the label rather than modeled, and its composition shifts with the refinancing environment across cohorts. Sharing the label across both arms makes every paired contrast well-defined, but does not by itself rule out a gain earned by tracking the non-default class's prepayment composition rather than credit risk---the macro series that condition the recurrence also drive refinancing---so Section~\ref{sec:fannie-macro} re-scores the headline contrast with early-prepaid loans excluded; absolute AUCs read against the fixed-horizon definition. A 1-in-10 loan subsample of the 2004Q1--2010Q4 acquisition quarters yields $1.04$M labelable loans (a seeded $300$k experiment subsample). The shift is as severe as deployment ever gets: the 24-month forward default rate runs $1.3\%$ for 2004 originations, peaks at $9.8\%$ for 2007, and collapses to $0.4\%$ by 2010.

Every design in this section was frozen by internal pre-registration before the corresponding fits ran: the base design, the breadth amendment, the macro-conditioning amendment, four follow-on studies, the Freddie Mac replication, and the pandemic-band episode. The freeze-commit audit trail is in Appendix~\ref{app:prereg-audit}; the amendment chronicle and its intermediate tables are in the Online Supplement (\S S3). Splits: \textsc{iid}; \textsc{vintage} (develop on originations through 2005, test on the 2007--2009 crisis cohort, 2006 excluded); secondary \textsc{short} (the $8.6\%$ of loans with incomplete report histories). Channels: three at the base design (delinquency depth, remaining-balance ratio, modification flag; oriented empirically per regime), widened to six by a pre-registered amendment after a label-free audit found the scheduled-versus-actual paydown axis to be the crisis-era archive's only other live content (a behind-schedule unpaid-principal-balance (UPB) gap against the level-payment annuity, its curtailment mirror, and the servicer's months-behind measure). Inference: paired-bootstrap CIs over an RNN-by-XGBoost seed grid, with a pre-specified five-seed ensembled downstream (unweighted probability mean, which preserves the guarantee) as the production protocol.

\subsubsection{Unconditioned recurrence does not survive the crisis shift}
\label{sec:fannie-uncond}

At three channels the pre-registered primary contrast---adding the three joint features to covariates plus supervised summaries---fails informatively. On \textsc{iid} it is zero-to-negative at every seed pair. On \textsc{vintage} it \emph{changes sign with the XGBoost seed}, with all four intervals individually clear of zero on opposite sides, because the shared baseline rung itself moves by $0.014$ AUC across downstream seeds under the crisis shift. Under a shift this severe, the paired-contrast methodology loses its footing: the quantity being contrasted is not seed-stable, so no per-seed interval licenses a direction. Seed-ensembling the downstream does not revive the claim---the sign still flips, now with the RNN seed. Neither does channel breadth: at six channels every cell sits within $\pm 0.002$ of zero on \textsc{iid} and \textsc{vintage}---the sign still seed-dependent inside that band, the magnitude no longer material (Table~\ref{tab:real-fannie-wide}).

What breadth does buy is the \emph{frame}. The six-channel supervised-summary baseline gains $+0.016$ AUC on the crisis cohort, and its cross-seed instability collapses from $0.014$ to $0.001$ AUC---channel breadth pays the governed frame first, consistent with the ladder. The one all-seeds-robust unconditioned gain is the \textsc{short} split ($+0.0039$ to $+0.0232$ across designs and seeds, CI-clear everywhere): on loans with gappy histories, masked recurrence aggregates what fixed windows cannot---the same population where per-channel recurrent features failed on Home Credit's card panel.

\subsubsection{Macro-conditioned recurrence clears the crisis bar}
\label{sec:fannie-macro}

The vintage shift is a macroeconomic regime change, and the exogenous conditioning of Section~\ref{sec:conditioning} is shape-free in exogenous context. A pre-registered amendment joins two public series per loan-month---the unemployment rate and the two-month-lagged twelve-month log-change of the national home-price index---and fits, per oriented channel, the macro-conditioned multihead bank of Section~\ref{sec:conditioning}: severity $c(z)$, decay-gate $a(z)$, and threshold $\tau(z)$ sub-banks, plus a macro-gated leaky-max $b(z)$ and the pinned HWM channel, against a matched static bank that isolates conditioning from capacity. Six per-channel scores enter the ladder as a $+1$-constrained block.

\emph{The decision rule is met---the first CI-clear crisis-cohort uplift from a learned block in any of this paper's mortgage studies} (the six-channel frame gain of Section~\ref{sec:fannie-uncond} is engineered content, not a learned-feature contrast) (Table~\ref{tab:real-fannie-macro}). Under the ensembled downstream, the primary contrast clears zero at every one of five RNN seeds ($+0.015$ to $+0.021$ AUC on the crisis cohort; random-effects pooled $+0.0170$ $[+0.0139, +0.0200]$, $\hat\tau = 0.002$, Table~\ref{tab:seed-meta}). The base macro amendment pre-registered RNN seeds $0$ and $1$; the three further seeds are a post-hoc hardening under the same frozen protocol. The conditioning contrast against the static bank clears at every seed ($+0.015$ to $+0.020$; the static bank alone contributes $-0.001$ to $+0.002$ across seeds, with no stable sign), so conditioning---not capacity---carries the gain. The block stays additive on top of the joint features in this episode (the pandemic band reverses that; Section~\ref{sec:fannie-pandemic}). The pre-declared extrapolation risk did not bind: gates trained on $4.4$--$5.6\%$ unemployment and rising prices generalized to $10\%$ unemployment and $-13\%$ year-over-year prices. Nor is the gain a prepayment-composition artifact of the shared label: one third of the crisis test cohort terminates by prepayment inside the label window ($33.1\%$, essentially all labeled non-default), and re-scoring the same ensembled scores with those loans excluded---an evaluation-only sensitivity, no refit---\emph{raises} the primary contrast at every seed, to $+0.019$ to $+0.026$ AUC on the $90{,}009$ retained loans ($6.4\%$ default rate, every interval CI-clear): the conditioning gain concentrates in loans that stay in the book. Table~\ref{tab:real-fannie-econ} (discussed in Section~\ref{sec:discussion}) translates the gain into fixed-cutoff approval economics on the same scores.

\subsubsection{Controls and follow-ups}
\label{sec:fannie-controls}

Four controls and follow-up studies sharpen the claim.

\paragraph{Replacement, confirmed on held-out seeds.}
Six conditioned scores \emph{replacing} the six-summary frame at equal feature budget clear zero on the crisis cohort and on \textsc{iid} at all three RNN seeds never touched by the exploratory analysis. The unconditioned bank as replacement loses decisively everywhere but \textsc{short}. Conditioning, not recurrence, carries the replacement (Supplement Table~S5).

\paragraph{Delivery control.}
This control is exploratory at the national level (selected after the macro result; its pre-registered form is the state-granularity delivery control below, which replicates it). Handing the same two series to the downstream model as eight free columns \emph{hurts} the frame under the shift ($-0.0020$, CI-clear negative---tree splits on macro levels frozen in the boom mis-extrapolate), while conditioning beats the $z$-augmented frame decisively. In distribution the ordering reverses: on the added-to-frame basis the free columns are worth $+0.0167$ against the conditioned block's $+0.013$ to $+0.014$ (point estimates), and on the replacement basis the conditioned block trails the $z$-augmented frame by $-0.007$, CI-clear. Figure~\ref{fig:conditioning-delivery} shows the asymmetry: the same two series, delivered two ways, sit on opposite sides of zero under the shift. The conditioning advantage is shift-specific by measurement, not assumption.

\begin{figure}[tp]
\centering
\includegraphics[width=\linewidth]{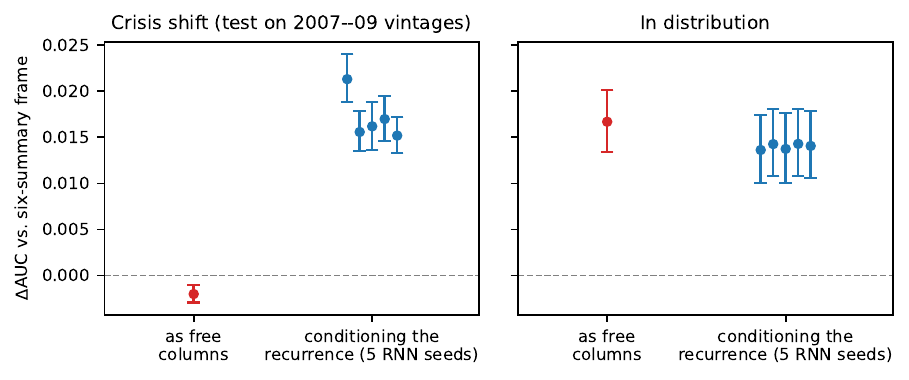}
\caption{The conditioning--delivery asymmetry. The same two macroeconomic series (unemployment, lagged home-price change) added to the six-summary frame two ways: as eight free input columns to the downstream model (red; RNN-independent) versus conditioning the recurrence of the monotone bank (blue; the five ensembled RNN seeds). Both panels plot the added-to-frame contrast; whiskers are paired-bootstrap $95\%$ CIs. Under the crisis shift (left) the two deliveries sit on opposite sides of zero; in distribution (right) the ordering reverses on point estimates (the CI-clear in-distribution reversal is on the replacement basis; Section~\ref{sec:fannie-controls}).}
\label{fig:conditioning-delivery}
\end{figure}

\paragraph{Content control: macro series versus calendar drift (exploratory).}
The between-cohort mechanism invites a deflationary reading: cohorts are ordered in time, so perhaps any smooth exogenous drift handle---calendar time itself---would deliver the same transfer, making the specific macro content incidental. An architecturally identical bank conditioned on a standardized $(t, t^2)$ calendar basis in place of the two macro series (same z-slot, dimensionality, and standardization frame; selected post hoc, run at both base RNN seeds under the ensembled protocol) answers no, in both directions. Under the crisis shift the calendar-conditioned bank is CI-clear \emph{negative} against the frame at both seeds ($-0.021$/$-0.046$) and against the static bank ($-0.023$/$-0.047$), while the macro bank beats it by $+0.042$/$+0.062$: a drift gate learned on the boom extrapolates monotonically into test vintages whose realized risk ordering peaks at 2007 and falls thereafter, so conditioning on time is worse than not conditioning at all---the same mis-extrapolation the delivery control exposes for downstream tree splits, reproduced inside the recurrence. In distribution the ordering reverses: the calendar bank clears ($+0.013$/$+0.015$) and is practically indistinguishable from the macro bank ($-0.0012$/$+0.0011$, opposite signs at the two seeds), because with cohorts mixed in support a calendar handle is just an in-support cohort index---the same face the modern-band \textsc{iid} cell shows. The crisis gain is therefore attributable to the macro \emph{content} of the conditioning, not to the presence of an exogenous handle; the guarantee is indifferent to the substitution, and the deployed checks record zero violations.

\paragraph{Mechanism, established by a pre-registered geography test.}
Replacing the national index with each loan's state-level series is the one intervention that could convert regime-level calibration into within-cohort borrower ranking. It leaves the pooled crisis gain intact, but the pre-specified within-cohort criterion straddles zero at both seeds, with the state increment a tight zero under the shift and CI-clear in distribution (Supplement Table~S6). The crisis gain \emph{is} between-cohort calibration transfer: what an exogenous regime handle should be expected to buy, and what post-hoc recalibration cannot mimic (Section~\ref{sec:discussion}).

\paragraph{Short histories.}
The conditioned block is worth $+0.055$ on \textsc{short}---calendar-faithful gap decay under the true macro regime (Section~\ref{sec:conditioning}) is exactly what incomplete histories were missing.

\subsubsection{Scope cells}
\label{sec:fannie-scope}

Two further pre-registered cells bound the claim (Table~\ref{tab:real-fannie-bands}). On a modern mild-shift band (develop 2015, test 2017; overlap with the CARES---Coronavirus Aid, Relief, and Economic Security---Act forbearance window pre-declared), the same block is worth a quarter of its crisis size on the vintage analog ($+0.0048$/$+0.0055$, CI-clear). On the modern \textsc{iid} split, where cohorts mix, it inflates to $+0.053$/$+0.055$ as an in-support cohort index---the between-cohort mechanism showing its other face. A dev-cutoff sweep (2004/2005/2006 against the fixed crisis test) finds the primary contrast CI-clear positive in all six cutoff-by-seed cells, so the result is not an artifact of the pre-registered boundary. Magnitude peaks at the anchor: the 2006 pool has the crisis onset in support (less to add), and the 2004 single-cohort pool spans too little macro variation to learn conditioning from. Conditioning needs regime variation in its training range before it can pay under regime shift, and pays crisis-sized only under a crisis---where \emph{crisis} means a realized credit-outcome regime shift, a macro shock that reaches default outcomes, not macro-series severity alone (the pandemic band of Section~\ref{sec:fannie-pandemic} makes that distinction sharp).

\subsubsection{Confirmatory population replication on Freddie Mac}
\label{sec:fannie-freddie}

The macro-conditioning result rests on one dataset, so a final pre-registered amendment (frozen, committed, and pushed before any Freddie fit; Appendix~\ref{app:prereg-audit}) replicates the confirmatory set on the Freddie Mac single-family archive: same band, unit, label, eligibility, splits, channels, macro series, and producers, on an independently collected loan population. Scope is declared in the amendment: the two GSEs' 2004--2010 books share the crisis, so this is a \emph{population} replication within the same macro episode---across loan populations and data-collection pipelines, not across regimes. Two channels required reconstruction from Freddie's thinner performance file (the servicer months-behind measure is rebuilt in closed form from published fields; a pre-registered five-channel sensitivity arm drops it entirely). The decision bar is the strictest used in this paper: the ensembled contrast must clear zero on the crisis cohort at \emph{all five} RNN seeds, every seed confirmatory since none was pre-seen.

\emph{The macro replication rule is met on both channel arms} (Table~\ref{tab:real-freddie}). The primary contrast clears on the crisis cohort at all five seeds ($+0.006$ to $+0.013$ AUC; five-channel arm $+0.005$ to $+0.011$; random-effects pooled $+0.0102$ $[+0.0072, +0.0133]$, Table~\ref{tab:seed-meta}). The conditioning contrast against the static bank clears at all five on both arms, and \textsc{iid} ($+0.013$) and \textsc{short} ($+0.041$ to $+0.046$) are CI-clear in every cell. Magnitudes are roughly half to two-thirds of Fannie's, on a base frame of comparable strength (\textsc{iid} base AUC $0.890$ vs.\ Fannie's $0.885$), with per-coordinate bump tests recording zero violations across all thirty joint fits.

\emph{The replacement rule is not met}, and the failure is population-sensitive in an informative way. At equal feature budget the conditioned block's crisis-cohort contrast changes sign with the RNN seed on Freddie (one seed clearly negative, three straddling, one clearly positive on the six-channel arm; the five-channel arm is likewise sign-unstable, three seeds clearly negative and one clearly positive), where it had confirmed at every held-out seed on Fannie. Everywhere else it replicates Fannie's shape: clear positive on \textsc{iid} at all five seeds, decisively positive on \textsc{short}, the unconditioned control decisively negative everywhere but \textsc{short}. The additive claim---conditioning \emph{added to} the governed frame pays under crisis shift at every seed---is therefore what replicates across GSE populations; the stronger replacement form does not transfer, and is bounded accordingly.

\subsubsection{A second macro episode: the pandemic band}
\label{sec:fannie-pandemic}

Every macro verdict above rests on the 2007--2009 episode, so a final amendment (frozen before any pandemic fit; Appendix~\ref{app:prereg-audit}) extends the design to the archive's only other severe macroeconomic event. Development pool: 2015--2016 originations, whose label windows close by December 2019, entirely pre-CARES. Test: 2020--2021 originations. Origination years 2017--2019 are excluded outright because their label windows cross the 2020Q2 CARES boundary---a three-year buffer forced by the label horizon, not chosen. Everything else is the crisis design unchanged (subsample rule, six channels, macro series, producers, ensembled downstream), run on both GSEs with the same two-contrast decision rule at RNN seeds $\{0,1\}$ and a pre-registered extension to seeds $\{2,3,4\}$ if met. The episode is structurally unlike 2008, and the amendment declares the measurement caveat central: unemployment spiked to $14.8\%$ and reversed within months, house prices boomed, and CARES forbearance changed how distress is recorded, so the ever-90-plus-days-delinquency label under this regime is not measured as in any other band (direction of distortion unassumed). The declared readouts locate that distortion. Per-vintage default rates are flat ($0.83$--$0.89\%$ in the development vintages, $0.68\%$ and $1.00\%$ in the 2020 and 2021 test cohorts---no crisis-style surge), and forbearance-family workout codes appear in $7.2\%$ of 2020Q1-origination feature windows, decaying to ${\sim}0.5\%$ by the 2021 vintages ($1.2\%$ of test-cohort loans pooled), so feature-window contamination concentrates in the earliest pandemic cohort.

\emph{The rule is met on Fannie Mae, at modern-band size, and the extension holds at all five seeds} (Table~\ref{tab:real-fannie-pandemic}). On the pandemic cohort the primary contrast is $+0.0044$ to $+0.0050$ AUC (random-effects pooled $+0.0048$ $[+0.0039, +0.0057]$, $\hat\tau = 0$, Table~\ref{tab:seed-meta}), the conditioning contrast $+0.0036$ to $+0.0041$ (pooled $+0.0039$ $[+0.0030, +0.0048]$), and the additive joint-feature secondary $+0.0028$ to $+0.0041$---every pandemic-cohort and \textsc{iid} cell CI-clear; the descriptive \textsc{short} cells are mostly clear with three straddles. Two magnitude readings matter. The gain is a quarter to a third of the crisis size; and, unlike the crisis, it is no larger than the same band's in-distribution gain ($+0.0066$ to $+0.0080$)---under flat default prevalence and a transient spike, conditioning pays roughly its in-support value.

The pre-declared per-origination-year readout says the same thing mechanistically: within-cohort gains ($+0.0031$ to $+0.0047$ inside the 2020 vintages, $+0.0053$ to $+0.0060$ inside 2021, from the persisted ensembled scores) bracket the pooled contrast, so this episode's gain lives \emph{within} cohorts---the in-support face of conditioning---where the crisis gain was between-cohort transfer with mixed within-year rungs. One computed contrast runs the other way and is reported for completeness: macro-over-joint, additive in the crisis, is a tight zero-to-negative here (CI-clear negative at one seed, $-0.0009$; straddling at the other four)---in this episode the conditioned bank and the unconditioned joint features are substitutes, not complements. Both readings strengthen the scope sentence of Section~\ref{sec:fannie-scope}.

\emph{On Freddie Mac the rule is not met, informatively.} The conditioning contrast clears at both seeds ($+0.0046$ at each), but the primary straddles zero at both ($+0.0009$ and $+0.0024$): on this population the static bank is CI-clear \emph{negative} against the pandemic-cohort frame ($-0.0036$/$-0.0021$), and conditioning only wins back what the bank costs. What transfers across GSEs in this episode is conditioning-over-capacity; the block-over-frame primary, which transferred cleanly in the crisis, is population-sensitive here---and both GSEs' verdicts, the Fannie positive as much as the Freddie null, read under the same CARES-measured label. (Freddie's first-payment-derived origination assignment conservatively trims the band's edge vintages, as in the crisis replication; realized cohort sizes accompany the artifacts. The free-columns delivery arm was not run in this band, so the conditioning--delivery asymmetry is evidenced in the crisis episode only.)

\begin{table}[tp]
\centering
\caption{Fannie Mae six-channel study (pre-registered breadth amendment; paired $\Delta$AUC of the primary contrast, 95\% CIs, full seed grid (seed pairs written XGB/RNN) and 5-seed ensembled downstream; Section~\ref{sec:fannie-uncond}).}
\label{tab:real-fannie-wide}
\footnotesize
\setlength{\tabcolsep}{4pt}
\adjustbox{max width=\linewidth}{\input{tables/real_fannie_wide}}
\end{table}

\begin{table}[tp]
\centering
\caption{Fannie Mae macro-conditioned recurrence (pre-registered amendment; paired $\Delta$AUC, 95\% CIs). Top: the primary contrast (six macro-conditioned per-channel bank scores added to the six-summary frame) over the full single-seed grid (written XGB/RNN; unstable under the crisis shift, six of ten cells clear positive) and the 5-seed ensembled downstream. RNN seeds $0$/$1$ are the pre-registered pair; seeds $2$--$4$ are the post-hoc hardening under the same frozen protocol. Middle: against the matched static bank (ensembled). Bottom: on top of the joint features (ensembled).}
\label{tab:real-fannie-macro}
\footnotesize
\setlength{\tabcolsep}{4pt}
\adjustbox{max width=\linewidth}{\input{tables/real_fannie_macro}}
\end{table}

\begin{table}[tp]
\centering
\caption{Fannie Mae scope cells (2026-07-07 amendment; macro primary and conditioning contrasts, 5-seed ensembled downstream). Modern band: develop 2015 / test 2017 originations (CARES caveat pre-declared). Dev-cutoff sweep: cutoffs 2004--2006 against the fixed 2007--2009 crisis test on identical test rows.}
\label{tab:real-fannie-bands}
\footnotesize
\setlength{\tabcolsep}{4pt}
\adjustbox{max width=\linewidth}{\input{tables/real_fannie_bands}}
\end{table}

\begin{table}[tp]
\centering
\caption{Freddie Mac confirmatory replication (pre-registered amendment 2026-07-08; paired $\Delta$AUC, 95\% CIs, 5-seed XGBoost-ensembled downstream, all five RNN seeds confirmatory). Top blocks: the six-channel arm---macro primary, conditioning, and replacement contrasts with the static-bank control. Bottom blocks: the five-channel sensitivity arm (reconstructed months-behind channel dropped). Verdicts in Section~\ref{sec:fannie-freddie}.}
\label{tab:real-freddie}
\footnotesize
\setlength{\tabcolsep}{4pt}
\adjustbox{max width=\linewidth}{\input{tables/real_freddie}}
\end{table}

\begin{table}[tp]
\centering
\caption{Pandemic-band macro episode (pre-registered amendment 2026-07-11; paired $\Delta$AUC, 95\% CIs, 5-seed XGBoost-ensembled downstream). Vintage here is the pandemic shift: development pool 2015--2016, test 2020--2021 originations, 2017--2019 excluded because their label windows cross the CARES boundary. Blocks per GSE: the macro primary, the conditioning contrast against the static bank, and the additive joint-feature secondary. RNN seeds $0$/$1$ are the pre-registered pair; the Fannie extension to seeds $2$--$4$ is triggered by its rule being met (Section~\ref{sec:fannie-pandemic}); \textsc{short} is descriptive. Labels are measured under CARES forbearance (caveat pre-declared, direction unassumed).}
\label{tab:real-fannie-pandemic}
\footnotesize
\setlength{\tabcolsep}{4pt}
\adjustbox{max width=\linewidth}{\input{tables/real_fannie_pandemic}}
\end{table}

\begin{table}[tp]
\centering
\caption{Random-effects pool of the headline macro contrasts over the five-RNN-seed grid (ensembled downstream; DerSimonian--Laird $\hat\tau^2$ from within-seed SEs recovered from the paired-bootstrap CIs; modified Knapp--Hartung $95\%$ intervals on $t_4$ quantiles, the wider of the KH and classical half-widths). $\hat\tau$ is the between-seed SD; caveats and reading rules in Section~\ref{sec:real-setup}.}
\label{tab:seed-meta}
\footnotesize
\setlength{\tabcolsep}{4pt}
\adjustbox{max width=\linewidth}{\input{tables/real_seed_meta}}
\end{table}

\subsection{Home Credit Stability: the ladder under an official stability metric}
\label{sec:real-hcs}

The strictness ladder is the paper's central real-data figure, but on AMEX it is a single-dataset result scored by a single competition metric. The 2024 Home Credit ``Credit Risk Model Stability'' competition supplies an independent replication instrument with two properties AMEX lacks: thirteen monotone event channels (dated delinquency, overdue-amount, and debt records from previous applications and credit-bureau contracts, monthly-binned against the decision date with post-decision records excluded), and an official metric that prices exactly what this paper claims monotone summaries provide---$\text{stability} = \text{mean}_w(\text{gini}_w) + 88\cdot\min(0, \text{slope}) - 0.5\cdot\text{std}(\text{resid})$ over weekly cohort ginis, penalizing a deteriorating weekly trend. On a seeded $300$k-case subsample we rebuild the four-frame ladder (free $229$-column panel; sign-confident constraints on the $147$ dev-signable columns; those columns only; supervised summaries only) and add the three joint features to each frame, on a temporal split (develop through week $64$, test on weeks $65$--$91$) and an \textsc{iid} split, over the same $2{\times}2$ seed grid.

The ladder replicates, and the pre-registered decision rule is met on \textsc{iid} (Figure~\ref{fig:hcs-ladder}; Table~\ref{tab:hcs-ladder}). The joint-feature uplift is statistically zero in the free, confident, and strict frames on both regimes and both metrics---and in the summaries frame it clears zero \emph{on the official stability metric in all four seed pairs} ($+0.046$ to $+0.092$) and on AUC in all four ($+0.026$ to $+0.041$). At thirteen channels the summaries-frame uplift is an order of magnitude larger than AMEX's $+0.004\,M$ at fifteen, on one third the features---though that comparison crosses metrics and channel types (raw event channels against pre-engineered composites), so the clean breadth evidence remains the two-to-eight-channel absence against the thirteen-to-fifteen-channel presence.

The temporal regime adds the metric's own verdict on stability: the thirteen supervised summaries alone have a \emph{falling} weekly gini trend (base stability $-0.09$ to $-0.12$ despite an AUC near $0.69$), and adding the three joint features reverses the trend in every seed pair, worth $+0.22$ to $+0.30$ stability in point estimate. We report that reversal as directional support only: the $88{\times}$slope term makes the statistic's bootstrap intervals wide enough to straddle zero in all four pairs, while the corresponding AUC uplift clears zero in three of four ($+0.013$ to $+0.037$). The pre-registered prediction that the stability metric would favor the monotone features more than raw AUC holds in point estimate on the temporal split and is indistinguishable from the AUC-implied effect on \textsc{iid}.

The 5-seed downstream ensemble pre-specified for the Fannie Mae study strengthens rather than weakens this replication: the summaries-frame stability uplift stays clear of zero at both RNN seeds on \textsc{iid} ($+0.062$ and $+0.067$), and the temporal AUC uplift---three of four seed pairs in the single-seed grid---clears zero at both RNN seeds under the ensemble ($+0.016$ and $+0.030$), while the large temporal stability gains remain directional ($+0.26$ and $+0.29$).

The competition's late-submission scoring supplies what no internal split can: an externally adjudicated readout on hidden test weeks \emph{later than all training data}, scored by the official metric. We submitted both summaries-frame arms exactly as validated internally (the deployable pipeline reproduces the internal test scores to machine precision; dev-frozen orientation, scalars, and downstream; no refitting). The external result is a point estimate with no CI machinery, and it is \emph{inconsistent} with the internal picture: the base arm scores $0.260$ public / $0.209$ private, the joint arm $0.132$ / $0.044$---adding the three states \emph{degrades} external stability by $-0.13$ to $-0.17$. (Absolute levels sit far below full-panel leaderboard entries by design: both arms are thirteen-channel, summaries-frame pipelines; the readout is the paired difference, not the level.) The degradation is not a single-seed artifact: rebuilding the joint arm at the second RNN seed (the base arm is seed-independent and reproduces its score exactly) replicates the sign at roughly half the magnitude ($0.213$ / $0.137$, a $-0.05$ to $-0.07$ degradation), so the direction is seed-robust while its size is not.

A post-hoc per-week decomposition of the internal test window locates the mechanism. The joint features' stability gain is earned early: in weeks 65--77 the base frame's weekly gini falls at $-0.021$ per week and the states flatten the fall to $-0.011$. In the final stretch (weeks 78--91) the base frame \emph{recovers} ($+0.005$ per week) and the joint arms stay flat ($-0.003$ and $-0.001$ at the two RNN seeds). The trend that extrapolates into the hidden later weeks therefore favors the base arm, the $88{\times}$slope penalty compounds against the joint arm there, and the steeper public-to-private drop is consistent with exactly that. We report this asymmetry as the study's production-shift bottom line: the joint features' in-distribution content is real (four-of-four CI-clear on \textsc{iid}, externally unfalsified on AUC-like level differences), while their \emph{temporal stability} contribution is fragile beyond the training horizon---the same unconditioned-recurrence fragility the Fannie crisis test exposes through seed instability, on the one dataset where the repair that worked there (exogenous macro conditioning) is structurally unavailable because the loans' country is undisclosed.

\begin{figure}[t]
\centering
\includegraphics[width=\linewidth]{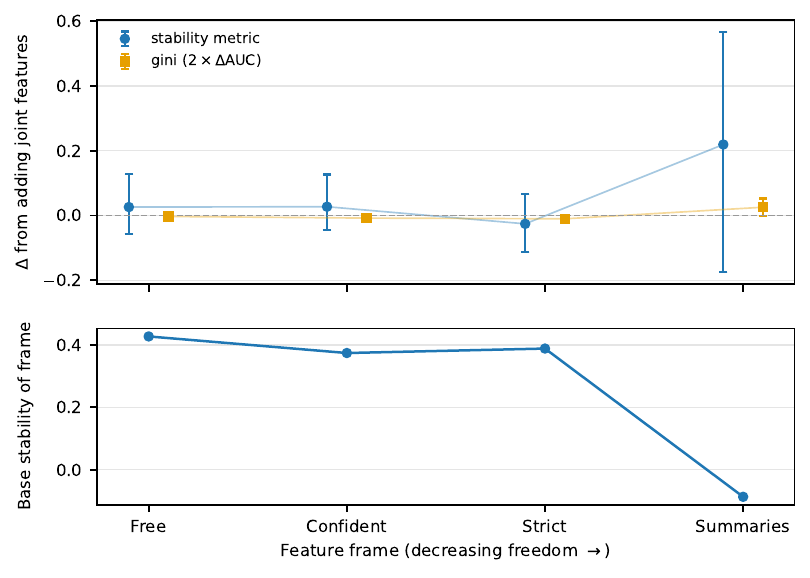}
\caption{The strictness ladder replicated on Home Credit Stability (temporal split, seed pair $0/0$). Top: uplift from adding the three joint features to each frame, under the official stability metric and under AUC rescaled to the gini scale. Bottom: base stability of each frame. As on AMEX, the uplift is zero while the frame is free and concentrates in the summaries frame.}
\label{fig:hcs-ladder}
\end{figure}

\begin{table}[tp]
\centering
\caption{Home Credit Stability ladder (seed pair $0/0$; base value of each frame and paired joint-feature contrast with 95\% CIs, under the official stability metric and raw AUC; bottom block: the pre-specified 5-seed XGBoost-ensembled downstream per RNN seed, stability metric). Seed-pair and ensemble verdicts in Section~\ref{sec:real-hcs}.}
\label{tab:hcs-ladder}
\footnotesize
\setlength{\tabcolsep}{4pt}
\adjustbox{max width=\linewidth}{\input{tables/hc_stability_ladder}}
\end{table}

\subsection{What the real-data studies establish}
\label{sec:real-synthesis}

Seven studies---the AMEX uplift, panel, and strictness-ladder contrasts, Home Credit, Fannie Mae, its Freddie Mac replication, and Home Credit Stability---plus the pandemic-band episode and the exploratory UCI probe combine into one conditional claim: learned cross-channel monotone recurrence earns a real, seed-robust accuracy premium precisely where deployment constraints bind, and under temporal regime shift it must be exogenously conditioned. This subsection assembles the claim from its parts; Table~\ref{tab:decision} compresses it for practice.

\paragraph{The ladder replicates and scales with breadth.}
The strictness-ladder pattern is not an AMEX artifact: it replicates on an independent dataset, at thirteen channels, under an official temporal-stability metric, with the summaries-frame uplift clear of zero in every seed pair (Section~\ref{sec:real-hcs}), while at Fannie Mae's three channels it is absent---exactly as the breadth boundary predicts. The Fannie breadth amendment adds a refinement: widening three channels to six moved the value into the admissible frame itself rather than into the joint-feature contrast. Channel breadth pays the governed frame first, and the learned recurrence only where histories are incomplete.

\paragraph{Unconditioned recurrence fails the crisis test.}
Under a $7{\times}$ default-rate swing, the paired baseline becomes seed-unstable, the contrast changes sign with the downstream seed, and neither five-seed ensembling nor doubled channel breadth revives it (Section~\ref{sec:fannie-uncond}). The Home Credit Stability leaderboard supplies the same verdict from the other side: on hidden weeks beyond the training horizon, the unconditioned joint states degrade the official stability score at both RNN seeds---on precisely the dataset where conditioning is structurally unavailable (Section~\ref{sec:real-hcs}). The two external tests agree on the boundary: unconditioned recurrence carries no trustworthy temporal signal beyond its training window.

\paragraph{Conditioned recurrence clears the bar, and the mechanism is identified.}
Macro-conditioned recurrence does survive: gates, severities, and thresholds driven by two published exogenous series clear the pre-registered bar on the crisis cohort at every seed, beating the matched static bank (Table~\ref{tab:real-fannie-macro}). The mechanism the shift breaks---a fixed notion of how long exposure history matters and how much a step counts---is exactly what conditioning repairs, at no cost to the guarantee (the gates are shape-free in exogenous context). The controls of Section~\ref{sec:fannie-controls} pin the attribution from three sides: the replacement form is confirmed at all held-out RNN seeds; the state-level test establishes the gain as between-cohort calibration transfer; and the content control shows the transfer requires the macro content itself. The conditioning-versus-delivery asymmetry replicates under pre-registration at state granularity.

\paragraph{The claim transfers across populations; its strongest form does not.}
On Freddie Mac's independently collected book of the same crisis, the additive conditioning contrast clears zero on the crisis cohort at all five RNN seeds on both channel arms; the stronger replacement form proves population-sensitive and is bounded accordingly (Section~\ref{sec:fannie-freddie}). Scope is likewise measured rather than assumed (Section~\ref{sec:fannie-scope}): the gain is robust to the development-cutoff choice, and it is crisis-sized only under a crisis---on a modern mild-shift band the same block is worth a quarter of its crisis value, while in distribution it inflates into an in-support cohort index, the between-cohort mechanism showing its other face.

\paragraph{What survives temporal stress without conditioning.}
Two results survive. The joint features reverse the falling weekly gini trend of the summaries frame on Home Credit Stability in every seed pair (directionally; the metric's slope penalty term is too high-variance for significance). And on loans with incomplete histories---the population where fixed-window aggregates degrade---they deliver Fannie Mae's one all-seeds-robust unconditioned gain, which conditioning more than doubles (Sections~\ref{sec:fannie-uncond}, \ref{sec:fannie-controls}).

\paragraph{The production story.}
The story this supports is conditional but concrete, and its in-distribution half is small in absolute terms: the ladder's top rung buys back about five percent of the summaries frame's own surrendered accuracy, and its temporal component should be presumed fragile unless conditioned. A governed scorecard restricted to admissible components can add three auditable learned features---each a monotone state with a one-line guarantee, versus $420$ hand-interaction columns that do not replicate them---for a gain that scales with channel breadth. An unconstrained pipeline with a mature engineered panel should expect no accuracy gain from the same features, and the guarantee itself---which the panel frame cannot offer---becomes the entire value proposition there. And any deployment that must cross a macroeconomic regime boundary should condition the recurrence on published exogenous series, delivered structurally rather than as columns.

\section{Discussion and Limitations}
\label{sec:discussion}

\paragraph{Scope of the evidence.}
The studies cover two consumer-credit panels, two GSE mortgage archives, one competition dataset with an official stability metric, and an exploratory single-channel probe on a small public card panel (UCI Taiwan). On this evidence, learned monotone recurrence pays when the admissibility frame binds and enough channels interact (the strictness ladder and the breadth boundary), and it does not pay against a free, saturated engineered panel. The guarantee is never the differentiator---it is equally available to engineered pipelines by composition---and the accuracy value of learning the aggregation is conditional on governance strictness and channel breadth. The mechanism-level complement, established under controlled synthetic conditions, is in the companion study \citep{instrument2026}. What remains outside both is any portfolio whose label mechanism, shift structure, or channel interactions differ materially from the datasets studied; and the severe-shift evidence covers two macroeconomic episodes---the 2007--2009 crisis (with the Freddie Mac result a population replication inside it) and the pandemic band (Section~\ref{sec:fannie-pandemic}), where the episode rule is met on one GSE population while on the other only the conditioning-versus-capacity contrast transfers.

\paragraph{How large are these effects?}
The paper's headline magnitudes span three metrics and are easiest to read against internal anchors. The AMEX joint-feature uplift is ${+}0.0044\,M$ in the summaries frame, whose own price relative to the free panel is large (${\sim}0.090\,M$): the three features buy back about five percent of that frame's surrendered accuracy. In the strict frame, whose price is the smaller ${\sim}0.014\,M$, the corresponding point estimate is ${+}0.0026\,M$---about a fifth of that frame's price, though its interval grazes zero and we do not claim it (Figure~\ref{fig:strictness}, Supplement Table~S1). The Home Credit Stability summaries-frame uplift ($+0.046$ to $+0.092$ on the official stability metric, every seed pair) is an order of magnitude larger at thirteen raw channels---and on the temporal split it separates a frame whose weekly discrimination \emph{trend} is falling from one that holds level. The Fannie macro-conditioning gain ($+0.015$ to $+0.021$ AUC on the crisis cohort, at all five RNN seeds) exceeds the entire cross-seed instability band ($0.014$ AUC) that sank the unconditioned claim, on a test cohort drawn from vintages spanning a $7{\times}$ default-rate swing; the same block is worth a quarter of that under a mild modern shift. For calibration, the companion synthetic study bounds the constraint's own predictive cost below $0.0017$ AUC on in-family mechanisms \citep{instrument2026}---below every effect claimed here, and an order of magnitude below the crisis-cohort conditioning gain (noting the AMEX and stability metrics are not on the AUC scale).

\paragraph{Economic translation.}
The uplifts reported here are ranking-metric differences, and the two internal anchors of the preceding paragraph govern their economic reading: adopting the three learned features does not repay the cost of governance, it reduces it (each frame's uplift is a fraction of that frame's own surrendered accuracy), and any gain claimed under the crisis shift must exceed the baseline's own $0.014$ AUC cross-seed band, which the conditioning gain does and the unconditioned contrast does not. For the headline crisis result, Table~\ref{tab:real-fannie-econ} translates the ranking gain into fixed-cutoff economics on the crisis test cohort ($134{,}556$ loans, $4.4\%$ default rate), using the same scored files behind Table~\ref{tab:real-fannie-macro} (reproduced exactly before translation). At an $80\%$ approval rate the conditioned score re-decides $3.4$--$4.8\%$ of the cohort, and the loans it swaps in default at $52$--$67\%$ of the rate of the loans it swaps out, per seed ($3.3$--$4.4\%$ against $5.5$--$7.2\%$): $9$--$21$ net defaults avoided per ten thousand approvals and $10$--$27$ basis points of defaulted unpaid balance removed from the approved book, positive at all five RNN seeds. The gain attenuates as the approval rate rises---the net-avoidance range already straddles zero at $90\%$, and at $95\%$ the swap set shrinks to ${\sim}1\%$---so the economic value concentrates where the decline region is deep enough for re-ranking to matter. These are descriptive quantities on the scored cohort: the seed range is the only uncertainty attached, no bootstrap intervals accompany the swap metrics, and the swap sets themselves are majority-stable across seeds (pairwise Jaccard overlap $0.63$--$0.81$ at the $80\%$ cutoff), so the re-decided loans are substantially the same loans at every seed. Two pieces of a full business case remain outside these numbers: they price defaults avoided at par (no loss-given-default or funding assumptions), and they take the score as a ranking, not a calibrated probability of default; cutoff placement in production would be set on a calibrated scale.

\begin{table}[tp]
\centering
\caption{Economic translation of the crisis-cohort conditioning gain (Fannie Mae vintage test cohort: $134{,}556$ loans, $4.4\%$ default rate). Fixed approval-rate cutoffs on the same ensembled scores behind Table~\ref{tab:real-fannie-macro}: the lowest-scored fraction of the cohort is approved. Frame rows are RNN-seed-independent; $+$macro and swap rows report the min--max over the five RNN seeds. The swap set is the loans the two models decide differently (equal size by construction); basis points are default-weighted unpaid principal balance in the approved book.}
\label{tab:real-fannie-econ}
\footnotesize
\setlength{\tabcolsep}{5pt}
\adjustbox{max width=\linewidth}{\input{tables/real_fannie_econ}}
\end{table}

\paragraph{Governance value beyond accuracy.}
Two governance properties are independent of the accuracy verdict. First, parsimony has an operating-cost reading: under SR 11-7-style validation, every admitted feature carries its own conceptual-soundness documentation, ongoing monitoring, and change control, so replacing hundreds of engineered columns (the hand-interaction control alone spans $420$) with three provable states---each a monotone function with a one-line guarantee and plottable transforms---reduces a recurring validation workload, not just a parameter count. The upstream producer is of course itself a model in the validation inventory, with its own training-data governance and change control; the arithmetic above nets that cost against the per-component one-line proofs and applies to the three-state AMEX and Home Credit Stability deployments---the mortgage deployment consumes six per-channel bank scores, a larger (still admissible) upstream surface. Second, the upstream-score deployment shape mirrors the sub-score stacking already standard in credit pipelines (behavioral scores and trended-data attributes feeding downstream models), so adopting the learned channel changes what the upstream component is, not where it sits.

\paragraph{Admissibility is necessary, not sufficient.}
The proofs of Section~\ref{sec:method} answer the monotonicity requirement; they do not close the deployment file. A macro-conditioned score is point-in-time by construction \citep{aguais2004pit}---its calendar movement at fixed borrower behavior is exactly the mechanism that pays under shift---and that movement carries obligations this paper does not evaluate: adverse-action reason codes must stay stable and explainable when the score moves with the regime rather than the borrower; state-level conditioning adds a regional axis for fair-lending monitoring; a deliberately through-the-cycle rating policy may prefer the static bank at the measured accuracy price. The proofs settle the monotonicity question of model validation, not the whole of model governance.

\paragraph{A practitioner's decision table.}
Table~\ref{tab:decision} compresses the empirical answer of these studies to ``should a governed pipeline adopt learned monotone recurrent features?'' into its three observable coordinates: how strictly the consuming frame is governed, how many interacting exposure channels the portfolio carries, and what kind of shift deployment must survive. Each cell cites the table or figure that measured it.

\begin{table}[tp]
\centering
\caption{When learned monotone recurrent features pay: expected value by governance frame, channel breadth, and shift type, with the evidence for each cell. ``Conditioned'' means exogenous context conditioning the recurrence (gates, severities, thresholds); magnitudes are $\Delta M$ (AMEX), stability (Home Credit Stability, HCS), or $\Delta$AUC (others). Tables S1, S2, and S4 reside in the Online Supplement.}
\label{tab:decision}
\scriptsize
\setlength{\tabcolsep}{3pt}
\begin{tabular}{p{0.24\linewidth}p{0.13\linewidth}p{0.17\linewidth}p{0.22\linewidth}p{0.13\linewidth}}
\toprule
Governance frame & Breadth & Shift & Expected value & Evidence \\
\midrule
Free engineered panel & any & any & none (guarantee only) & Tab.~\ref{tab:real-amex-panel}, Fig.~\ref{fig:strictness} \\
All-signed strict frame & $13$--$15$ & in-distribution & boundary at $15$ (${+}0.002$--$0.003\,M$); zero at $13$ & Suppl.\ Tab.~S1; Tab.~\ref{tab:hcs-ladder} \\
Summaries-only frame & $13$--$15$ & in-distribution & CI-clear (${+}0.004\,M$; ${+}0.05$--$0.09$ stability) & Tab.~\ref{tab:real-amex-uplift}, \ref{tab:hcs-ladder} \\
Summaries-only frame & $2$--$8$ & in-distribution & zero (breadth boundary) & Tab.~\ref{tab:real-hc}, \ref{tab:real-fannie-wide} \\
Uncompressed aggregates, all signed & $15$ & in-distribution & CI-clear (${+}0.005$--$0.009\,M$); zero if volatility columns stay sign-free & Sec.~\ref{sec:real-strictness} \\
Single channel in isolation (exploratory) & $1$ & in-distribution & unreliable: sign is a channel property ($\pm 0.02$ AUC, CI-clear both ways) & Suppl.\ Tab.~S2 \\
Mortgage design (masked recurrence, gap decay) & $3$--$6$ & short/gappy histories & robust (${+}0.004$--$0.023$; ${+}0.055$ conditioned); negative on the card panel's history split & Tab.~\ref{tab:real-fannie-wide}, \ref{tab:real-fannie-macro}; contra \ref{tab:real-hc} \\
Governed frame, unconditioned & any & temporal regime shift & do not trust (sign-unstable; degrades externally) & Suppl.\ Tab.~S4, Sec.~\ref{sec:real-hcs} \\
Governed frame, conditioned & $6$ & temporal regime shift & ${+}0.015$--$0.021$ Fannie (seeds $2$--$4$ post-hoc); ${+}0.006$--$0.013$ Freddie (pre-registered) & Tab.~\ref{tab:real-fannie-macro}, \ref{tab:real-freddie} \\
Governed frame, conditioned & $6$ & pandemic-era shift (forbearance-distorted label) & ${+}0.004$--$0.005$ Fannie (all five seeds, pre-registered rule met); Freddie primary straddles, conditioning clears & Tab.~\ref{tab:real-fannie-pandemic} \\
\bottomrule
\end{tabular}
\end{table}

\paragraph{Objections from practice.}
Two referee-shaped objections deserve answers from controls already run.

\paragraph{Objection: the macro gain is just data sourcing.}
\emph{``Banks already source external and bureau data.''} Bureau attributes are borrower-level; the macro handle is regime-level, and the macro-direct control \emph{is} the sourcing scenario run to completion: the same two public series handed to the downstream model as columns---exactly how sourced data enters a scorecard---hurt the governed frame under the crisis shift while conditioning the recurrence on them clears the bar, a result replicated under pre-registration at state granularity (Section~\ref{sec:fannie-controls}, Figure~\ref{fig:conditioning-delivery}). The content control closes the objection from the other side: conditioning on calendar time instead of the macro series is CI-clear negative under the shift, so it is the published macro content, not the conditioning machinery, that pays. In distribution the objection is \emph{correct}: the free columns beat the conditioned encoding there, which is why the claim is scoped to governed frames under shift rather than asserted unconditionally.

\paragraph{Objection: calibration transfer is just recalibration.}
\emph{``Score alignment across cohorts is standard practice.''} Post-hoc recalibration maps scores to outcomes using labels from the regime being calibrated to; in the deployment this paper simulates, those labels are two years in the future when the model ships. Conditioning transfers the calibration \emph{ex ante}, from published exogenous series available contemporaneously---the distinction is not the size of the adjustment but when it can be made, and the modern-band cell (a quarter of the crisis-sized gain under a mild shift) prices how much of it a benign decade would have left on the table.

\paragraph{Limitations and future work.}
Five limitations bound the claims. (i) The severe-shift evidence covers two macroeconomic episodes, and the second is measured under CARES forbearance (label distortion pre-declared, direction unknown) with a verdict split across GSE populations; the archive contains no third severe episode, so episode generality beyond these two is untested. (ii) The all-seeds-clear criterion is a robustness standard, not a familywise-error-controlled test; the random-effects pool of Table~\ref{tab:seed-meta} puts the seed axis on formal footing, but its interval shares test rows across seeds and is anti-conservative on that component. Beyond the seed axis, the confirmatory record now spans seven sequential amendments carrying roughly a dozen decision rules across two GSEs and multiple regimes, each adjudicated by uncorrected per-contrast intervals; no across-family correction is applied, and the sequential, disclosed structure of the amendments (each motivated by the preceding verdict; outcomes enumerated in Table~\ref{tab:prereg-audit}) is the mitigation offered rather than a formal one. A crude family-level read is nonetheless available: recovering the implied standard errors from the pooled intervals of Table~\ref{tab:seed-meta}, every headline pooled contrast---primary and conditioning on the shifted cohorts, both episodes, both GSEs---carries a $t$-statistic of at least $9.3$, against a Bonferroni-corrected critical value of about $5.9$ at $\alpha = 0.05/12$ on $t_4$, so the headline verdicts survive correction across the full rule family with room to spare. (iii) Pre-registration was internal (version-frozen, publicly auditable; the freeze evidence and two disclosed timestamp caveats are laid out in Appendix~\ref{app:prereg-audit}) rather than external, so the recommendation to register externally stands unmet for future confirmatory work. (iv) The economic translation of Table~\ref{tab:real-fannie-econ} is fixed-cutoff and rank-based; expected-loss pricing under loss-given-default and calibration assumptions is not performed, and the prepayment sensitivity of Section~\ref{sec:fannie-macro} is an exclusion test, not a joint competing-risks model of default and prepayment. (v) The AMEX \textsc{time} instability and the Home Credit Stability external degradation both concern short-horizon temporal fragility of unconditioned recurrence; diagnosing whether the former is estimation noise in a small cohort or the same fragility is open, and the Stability dataset structurally forbids the conditioning repair, so that hypothesis is untestable there.

\section{Conclusion}
\label{sec:conclusion}

This paper asked what learned monotone temporal aggregation is worth inside a governed credit-scoring pipeline, and answered on production data at matched admissibility, with proofs that the vector-input and conditioning extensions cost nothing in admissibility. The answer is conditional, and the conditions are the contribution. The value of learned monotone recurrence is the price of the freedom the governance frame removes: statistically zero where the consuming model is free, maximal where only auditable summaries are admitted, scaling with channel breadth, and replicated across two production-scale datasets and an official temporal-stability metric. Under temporal regime shift the answer sharpens to a deployment rule. Unconditioned learned recurrence does not survive a genuine crisis shift; the same recurrence conditioned on published exogenous series clears the crisis-cohort bar at every seed under this paper's seed-robustness standard---confirmed at $+0.006$ to $+0.013$ AUC on the fully held-out Freddie Mac replication, with the larger Fannie Mae range the discovery estimate---beats the identical information delivered as free columns, survives the exclusion of early-prepaid loans, and transfers between-cohort calibration ex ante, before crisis labels exist, while a second, forbearance-distorted episode bounds the payoff: crisis-sized only when the macro shock reaches default outcomes. Learned monotone recurrence deployed across a temporal shift needs an exogenous handle on the regime, or its temporal contribution should not be trusted beyond the training window. The measurements, their boundaries, and the audit trail that disciplines them are what this paper offers to practice.

\FloatBarrier

\appendix

\section{Proofs for the vector-input extension}
\label{app:bases}

Throughout, $x \preceq x'$ denotes the coordinatewise order of Section~\ref{sec:setup} (applied per coordinate and per step under vector input).

\begin{lemma}[Low-water mark]
\label{lem:lwm}
For monotone non-decreasing $g$, the low-water mark $\ell_t = \min\big(\ell_{t-1},\, g(x_t)\big)$ with fixed $\ell_0$ is monotone non-decreasing in each $x_t$.
\end{lemma}
\begin{proof}
Identical to Lemma~\ref{lem:hwm} with $\min$ in place of $\max$: $\min(\cdot,\cdot)$ is non-decreasing in each argument.
\end{proof}

\begin{lemma}[Non-negative-weight MLPs are isotone on vectors]
\label{lem:vector-mlp}
A feed-forward network with elementwise non-negative weight matrices, arbitrary biases, elementwise monotone non-decreasing activations, and a final softplus is monotone non-decreasing as a map $\mathbb{R}^{C}\to\mathbb{R}^{J}_{\ge 0}$ under the coordinatewise order on both sides, and its outputs are non-negative.
\end{lemma}
\begin{proof}
An affine layer $u \mapsto Au + b$ with $A \ge 0$ elementwise preserves the coordinatewise order, elementwise non-decreasing activations preserve it, and compositions of order-preserving maps are order-preserving; nothing in this argument uses input or output dimension one. The final softplus is increasing and positive.
\end{proof}

\begin{proof}[Proof of Proposition~\ref{prop:vector}]
Fix an output coordinate $j$ and define the scalar sequence $u^{(j)}_t = g_j(x_t)$. By Lemma~\ref{lem:vector-mlp}, raising any coordinate $x_{t,i}$ of the path can only raise $u^{(j)}_t$, and leaves $u^{(j)}_s$ for $s \ne t$ unchanged. Each state channel is then exactly the scalar recurrence of Lemma~\ref{lem:ewma} (EWMA), Lemma~\ref{lem:hwm} (HWM), or Lemma~\ref{lem:lwm} (LWM) driven by the inputs $u^{(j)}_{1:T}$, hence non-decreasing in every $x_{t,i}$ by composition of non-decreasing maps. The monotone head composes as in Theorem~\ref{thm:score}. For the conditioned multihead bank, the same substitution applies to Theorem~\ref{thm:cond-bank}: each bank channel's increment is one coordinate of the Lemma~\ref{lem:vector-mlp} map, and every step of that proof is unchanged.
\end{proof}

\section{Proofs for exogenous conditioning}
\label{app:cond-proofs}

\begin{proof}[Proof of Lemma~\ref{lem:gated-ewma}]
Unrolling the recurrence (with the empty-product convention $\prod_{r=T+1}^T a_r = 1$) gives
\begin{equation*}
e_T = \Big(\prod_{r=1}^T a_r\Big) e_0 \;+\; \sum_{s=1}^T \Big(\prod_{r=s+1}^T a_r\Big) g_1(x_s).
\end{equation*}
By (D2) the gate is causal-current, so among these factors $x_k$ appears only in $a_k$ and in the increment $g_1(x_k)$. Collect the terms containing $a_k$ and factor it out; the $s=k$ increment carries coefficient $C_k := \prod_{r=k+1}^T a_r$; the rest is independent of $x_k$:
\begin{equation*}
e_T = a_k\, P + C_k\, g_1(x_k) + R,
\end{equation*}
where $P, C_k, R$ involve only $\{a_r : r \ne k\}$ and $\{g_1(x_s) : s \ne k\}$, hence are independent of $x_k$, and are each $\ge 0$ (sums of products of non-negative terms, using $a_r>0$, $e_0\ge0$, $g_1\ge0$). A direct identification of coefficients gives $P = C_k\, e_{k-1}$, so
\begin{equation*}
\frac{\partial e_T}{\partial x_k}
= C_k\Big[\, e_{k-1}\,\frac{\partial a_k}{\partial x_k} \;+\; g_1'(x_k) \,\Big] \;\ge\; 0,
\end{equation*}
since $C_k \ge 0$, $e_{k-1} \ge 0$ (a non-negative accumulant), $\partial a_k/\partial x_k \ge 0$ by (D2), and $g_1'(x_k) \ge 0$. Where derivatives do not exist (e.g., ReLU gates), the same conclusion follows directly: $e_T$ is an affine combination, with non-negative coefficients $P$ and $C_k$, of the two functions $a_k$ and $g_1(x_k)$, both non-decreasing in $x_k$. The edge cases $k=1$ (where $P=C_1 e_0$) and $k=T$ (where $C_T=1$, $R=0$) are covered by the empty-product convention.
\end{proof}

\paragraph{The direction condition is necessary.}
The requirement in (D2) that the gate be non-decreasing in $x_t$ cannot be dropped. Take $T=2$, $g_1=\mathrm{id}$, $e_0=0$, and a \emph{decreasing} gate $a_2 = \sigma(-x_2)$. Then $e_2 = \sigma(-x_2)\,x_1 + x_2$ and
\begin{equation*}
\frac{\partial e_2}{\partial x_2} = -\,\sigma'(-x_2)\,x_1 + 1,
\end{equation*}
which is negative whenever $x_1$ is large (e.g., $x_1=100$, $x_2\approx0$: $-25+1=-24<0$). A gate that decreases in a protected input therefore breaks monotonicity: raising a late exposure can shrink the retained weight on a large past by more than its own increment adds.

\paragraph{Shape-freedom of purely exogenous gates.}
If the gate depends only on the exogenous $(t, z_t)$ and not on any protected input, then $\partial a_k/\partial x_k = 0$ and the display above reduces to $\partial e_T/\partial x_k = C_k\, g_1'(x_k) \ge 0$ for \emph{any} shape of $a(t, z_t)$. Monotonicity in exposure then requires only $a_t \ge 0$; the upper bound $a_t < 1$ in (D1) is needed for contraction and boundedness, not for monotonicity. This licenses macro-conditioned decay with full functional freedom. Boundedness itself survives gating: for finite $T$ each product is at most $1$, so $e_T \le e_0 + \sum_{s} g_1(x_s)$, and if $\sup_t a_t \le \bar a < 1$ then $e_T \le e_0 + \max_s g_1(x_s)/(1-\bar a)$, recovering the geometric stability bound of the constant-$\alpha$ EWMA.

\begin{proof}[Proof of Lemma~\ref{lem:severity}]
By exogeneity, $c(z_t)$ is constant in every $x_k$, so $u_t = c(z_t)\, g(x_t)$ satisfies $u_t \ge 0$ and $\partial u_t/\partial x_t = c(z_t)\, g'(x_t) \ge 0$. Among the unrolled factors of the gated EWMA, $x_k$ appears exactly where $g(x_k)$ did, now scaled by the non-negative constant $c(z_k)$; the coefficients in the proof of Lemma~\ref{lem:gated-ewma} are unchanged up to that non-negative factor. For the leaky-max, $r_t = \max(\beta r_{t-1},\, u_t)$ is a max of non-negative, non-decreasing terms, and the composition argument of Lemma~\ref{lem:hwm} applies term by term. Non-negativity of $c$ is necessary, not cosmetic: it keeps $u_t \ge 0$, hence the accumulant $e_{t-1} \ge 0$ on which the gated-EWMA proof relies. On a gap step the increment is dropped, so severity is correctly absent there.
\end{proof}

\begin{proof}[Proof of Lemma~\ref{lem:threshold}]
For fixed $\tau$, $g(x;\tau)=\max(0,\, x-\tau)$ is non-decreasing in $x$ with slope in $\{0,1\}$ and is non-negative. Since $\tau(z_t)$ does not depend on any $x_s$, the per-step derivative in $x_t$ is unaffected by its value, so $\partial g/\partial x_t \in \{0,1\}$ independent of $\tau$'s shape. Substituting this increment into either channel preserves monotonicity by the composition in the proof of Lemma~\ref{lem:severity}. The condition $\tau(z_t) \ge 0$ is a semantic convention (keep the threshold inside the non-negative exposure range), not required for monotonicity or non-negativity.
\end{proof}

\begin{proof}[Proof of Lemma~\ref{lem:gated-leakymax}]
Unrolling from $r_0 = 0$ (so the $b_1 r_0 = 0$ branch is dominated by $g_2(x_1) \ge 0$) gives, by induction and with the empty-product convention $\prod_{s > T} = 1$,
\begin{equation*}
r_T = \max_{1 \le t \le T}\Big(\prod_{s=t+1}^T b(z_s)\Big)\, g_2(x_t).
\end{equation*}
Each term is a non-negative ($\prod b \ge 0$, $g_2 \ge 0$) multiple of $g_2(x_t)$, non-decreasing in $x_t$; a max of non-decreasing functions is non-decreasing. By exogeneity, $b$ is constant in $x_k$, so raising $x_k$ moves only the single factor $g_2(x_k)$ and no $\prod b$ coefficient, whence $\partial r_T/\partial x_k \ge 0$ for any shape of $b$. The constant-$\beta$ leaky-max is the member $b \equiv \beta$; a pinned $b \equiv 1$ channel is the never-forgotten HWM and is left ungated. On a gap step the update $r_t = b(z_t)\, r_{t-1}$ (no observed increment) is monotone by the same argument, and a pinned channel is unchanged, preserving the peak across the gap.
\end{proof}

\begin{proof}[Proof of Theorem~\ref{thm:cond-bank}]
Each gated or constant-decay EWMA channel is monotone in every $x_k$ by Lemma~\ref{lem:gated-ewma} (with severity- or threshold-conditioned increments handled by Lemmas~\ref{lem:severity} and~\ref{lem:threshold}); each gated or constant leaky-max channel, including the pinned HWM, is monotone by Lemma~\ref{lem:gated-leakymax}. The head $F$ is monotone non-decreasing in all channels, and compositions of order-preserving maps are order-preserving, exactly as in Theorem~\ref{thm:score}.
\end{proof}

\section{Pre-registration audit trail}
\label{app:prereg-audit}

The design file is \texttt{docs/fannie\_hc\_\-prereg.md} in the project repository, with each freeze evidenced by version-control history (the definition of internal pre-registration is in Section~\ref{sec:real-setup}). Table~\ref{tab:prereg-audit} maps each confirmatory design to its prereg section, freeze commit, and the commit and push timestamps (UTC); the amendment chronicle and the intermediate tables it cites are in the Online Supplement (\S S3). Two rows require disclosure. The 2026-07-07 amendment (follow-on studies A--D) was committed before its fits but pushed only the following day, so its external timestamp postdates some of those fits; the freeze evidence for that row is the local commit chain (the commit is an ancestor of every subsequently pushed head). The five-RNN-seed hardening of the macro result was \emph{not} pre-registered---the base file fixes RNN seeds $\{0,1\}$ and its ensemble amendment expands only the downstream seeds---and is therefore reported as a post-hoc seed-robustness expansion wherever it appears; the all-five-seeds decision bar first becomes a pre-registered primary in the Freddie Mac amendment, where every seed is genuinely held out. The two recorded protocol deviations (a tie-degenerate quantile rule for the short-history split replaced by the data-defined boundary, and the experiment-level subsample size) were written into the prereg file before the corresponding fits. The pandemic-band amendment carries two further dated pre-fit deviation notes: its declared external (OSF) registration was waived before any fit (so its registration status is internal, like every other design here), and a label-free degeneracy-guard floor was lowered to admit the valid $1.7\%$ Freddie short-history cohort. Its push postdates its decision-rule fits (the seed-extension fits postdate the push and so carry full external freeze evidence); the freeze evidence for the decision-rule fits is the local commit chain. Amendment names use local dates and commit timestamps are UTC (Table~\ref{tab:prereg-audit}).

\begin{table}[H]
\centering
\caption{Pre-registration audit trail: confirmatory designs, freeze commits, timestamps (UTC), and decision-rule outcomes (the chronicle, Supplement Table~S3, keeps the per-rule detail). Commits are in the public project repository; push timestamps are from the hosting service's event log.}
\label{tab:prereg-audit}
\scriptsize
\setlength{\tabcolsep}{3pt}
\begin{tabular}{p{0.28\linewidth}p{0.16\linewidth}lllp{0.16\linewidth}}
\toprule
Design & Prereg section & Commit & Committed & Pushed & Outcome \\
\midrule
Fannie base design: unit, label, splits, primary contrast, decision rule; HC Stability ladder design & base file (2026-07-04) & \texttt{416e172} & 07-06 03:57 & 07-06 16:02 & Fannie not met; HCS \textsc{iid} met \\
Seed-ensembled downstream protocol; six-channel breadth amendment (label-free audit) & 2026-07-06 & \texttt{6400008} & 07-06 12:27 & 07-06 16:02 & not met (contrast not revived) \\
Macro-conditioned recurrence: exogenous series, producer arms, decision rule & 2026-07-06b & \texttt{03be486} & 07-06 16:01 & 07-06 16:02 & met \\
Follow-on studies A--D: replacement promotion (held-out seeds), state-level HPI test, modern band, dev-cutoff sweep & 2026-07-07 & \texttt{df4b75a} & 07-07 14:15 & 07-08 04:58 & A met; B not met (qualifier survives); C--D scope as measured \\
Five-RNN-seed hardening of the macro result & \emph{not pre-registered} (post-hoc robustness expansion) & --- & --- & --- & met (post hoc) \\
Freddie Mac confirmatory replication: field mapping, reconstructed channel, sensitivity arm, all-five-seeds rule & 2026-07-08 & \texttt{812b0b1} & 07-08 04:58 & 07-08 04:58 & macro met (both arms); replacement not met \\
Pandemic-band macro episode, both GSEs: band, splits, CARES caveat, decision and extension rules & 2026-07-11 & \texttt{f6815fa} & 07-10 19:10 & 07-11 16:45 & Fannie met (ext.\ incl.); Freddie not met \\
\bottomrule
\end{tabular}
\end{table}

\input{main.bbl}
\end{document}

%% file: tables/real_amex_uplift.tex
%
\begin{tabular}{cccccccc}
\toprule
 &  & \multicolumn{2}{c}{IID} & \multicolumn{2}{c}{Time} & \multicolumn{2}{c}{Peak} \\
\cmidrule(lr){3-4}\cmidrule(lr){5-6}\cmidrule(lr){7-8}
RNN seed & XGB seed & $\Delta M$ & 95\% CI & $\Delta M$ & 95\% CI & $\Delta M$ & 95\% CI \\
\midrule
0 & 0 & $+0.0044$ & $[+0.0009, +0.0080]$ & $+0.0046$ & $[+0.0014, +0.0077]$ & $+0.0040$ & $[+0.0019, +0.0060]$ \\
1 & 0 & $+0.0043$ & $[+0.0012, +0.0072]$ & $+0.0009$ & $[-0.0029, +0.0048]$ & $+0.0026$ & $[+0.0007, +0.0045]$ \\
2 & 0 & $+0.0049$ & $[+0.0017, +0.0079]$ & $+0.0031$ & $[+0.0001, +0.0062]$ & $+0.0047$ & $[+0.0028, +0.0067]$ \\
0 & 1 & $+0.0044$ & $[+0.0013, +0.0076]$ & $+0.0070$ & $[+0.0042, +0.0099]$ & $+0.0032$ & $[+0.0012, +0.0051]$ \\
0 & 2 & $+0.0044$ & $[+0.0015, +0.0073]$ & $+0.0075$ & $[+0.0046, +0.0106]$ & $+0.0033$ & $[+0.0014, +0.0053]$ \\
0 & 3 & $+0.0051$ & $[+0.0018, +0.0083]$ & $+0.0060$ & $[+0.0028, +0.0089]$ & $+0.0037$ & $[+0.0016, +0.0059]$ \\
\bottomrule
\end{tabular}

%% file: tables/real_amex_controls.tex
%
\begin{tabular}{lcccccc}
\toprule
 & \multicolumn{2}{c}{IID} & \multicolumn{2}{c}{Time} & \multicolumn{2}{c}{Peak} \\
\cmidrule(lr){2-3}\cmidrule(lr){4-5}\cmidrule(lr){6-7}
Contrast & $\Delta M$ & 95\% CI & $\Delta M$ & 95\% CI & $\Delta M$ & 95\% CI \\
\midrule
states${}+{}$joint (48) $-$ supervised summaries (15) & $+0.0042$ & $[+0.0006, +0.0077]$ & $-0.0119$ & $[-0.0158, -0.0082]$ & $+0.0072$ & $[+0.0048, +0.0096]$ \\
states${}+{}$joint (48) $-$ product aggregates (420) & $+0.0022$ & $[-0.0013, +0.0055]$ & $-0.0023$ & $[-0.0060, +0.0011]$ & $+0.0058$ & $[+0.0034, +0.0082]$ \\
product aggregates (420) $-$ supervised summaries (15) & $+0.0020$ & $[-0.0011, +0.0052]$ & $-0.0095$ & $[-0.0132, -0.0061]$ & $+0.0014$ & $[-0.0008, +0.0038]$ \\
\bottomrule
\end{tabular}

%% file: tables/real_amex_panel.tex
%
\begin{tabular}{lcccccc}
\toprule
 & \multicolumn{2}{c}{IID} & \multicolumn{2}{c}{Time} & \multicolumn{2}{c}{Peak} \\
\cmidrule(lr){2-3}\cmidrule(lr){4-5}\cmidrule(lr){6-7}
 & $\Delta M$ & 95\% CI & $\Delta M$ & 95\% CI & $\Delta M$ & 95\% CI \\
\midrule
Panel base $M$ (188 features) & \multicolumn{2}{c}{$0.7829$} & \multicolumn{2}{c}{$0.6485$} & \multicolumn{2}{c}{$0.6352$} \\
\midrule
$+$ all 48 RNN features & $+0.0011$ & $[-0.0020, +0.0041]$ & $+0.0020$ & $[-0.0009, +0.0047]$ & $-0.0012$ & $[-0.0033, +0.0009]$ \\
$+$ 3 joint features & $+0.0010$ & $[-0.0022, +0.0041]$ & $+0.0032$ & $[+0.0003, +0.0063]$ & $-0.0010$ & $[-0.0032, +0.0013]$ \\
$+$ 15 supervised summaries & $+0.0008$ & $[-0.0018, +0.0034]$ & $+0.0048$ & $[+0.0024, +0.0071]$ & $-0.0011$ & $[-0.0031, +0.0009]$ \\
$+$ 3 joint features given summaries & $-0.0011$ & $[-0.0038, +0.0015]$ & $+0.0000$ & $[-0.0030, +0.0029]$ & $-0.0003$ & $[-0.0023, +0.0016]$ \\
\bottomrule
\end{tabular}

%% file: tables/real_hc_ladders.tex
%
\begin{tabular}{lcccc}
\toprule
 & \multicolumn{2}{c}{IID} & \multicolumn{2}{c}{Short-history} \\
\cmidrule(lr){2-3}\cmidrule(lr){4-5}
Contrast & $\Delta$AUC & 95\% CI & $\Delta$AUC & 95\% CI \\
\midrule
\multicolumn{5}{l}{\emph{Card panel (8 channels)}} \\
hand aggregates $-$ covariates & $+0.0228$ & $[+0.0148, +0.0308]$ & $+0.0224$ & $[+0.0190, +0.0258]$ \\
$+$ per-channel RNN scalars & $-0.0008$ & $[-0.0033, +0.0017]$ & $-0.0031$ & $[-0.0045, -0.0018]$ \\
$+$ per-channel RNN states & $+0.0013$ & $[-0.0011, +0.0040]$ & $-0.0022$ & $[-0.0040, -0.0004]$ \\
$+$ joint states & $+0.0011$ & $[-0.0014, +0.0036]$ & $-0.0006$ & $[-0.0020, +0.0009]$ \\
joint $-$ per-channel states & $-0.0003$ & $[-0.0028, +0.0021]$ & $+0.0017$ & $[+0.0001, +0.0033]$ \\
\midrule
\multicolumn{5}{l}{\emph{Installments panel (2 raw timing channels)}} \\
hand aggregates $-$ covariates & $+0.0094$ & $[+0.0067, +0.0124]$ & $+0.0069$ & $[+0.0051, +0.0087]$ \\
$+$ per-channel RNN states & $+0.0005$ & $[-0.0006, +0.0017]$ & $+0.0001$ & $[-0.0006, +0.0008]$ \\
$+$ joint states & $-0.0000$ & $[-0.0012, +0.0012]$ & $-0.0000$ & $[-0.0006, +0.0006]$ \\
$+$ states and joint & $+0.0002$ & $[-0.0009, +0.0014]$ & $+0.0005$ & $[-0.0002, +0.0012]$ \\
\bottomrule
\end{tabular}

%% file: tables/real_fannie_wide.tex
%
\begin{tabular}{lcccccc}
\toprule
 & \multicolumn{2}{c}{IID} & \multicolumn{2}{c}{Vintage OOD (crisis)} & \multicolumn{2}{c}{Short-history} \\
\cmidrule(lr){2-3}\cmidrule(lr){4-5}\cmidrule(lr){6-7}
Primary contrast & $\Delta$AUC & 95\% CI & $\Delta$AUC & 95\% CI & $\Delta$AUC & 95\% CI \\
\midrule
Base AUC (covariates $+$ 6 summaries) & \multicolumn{2}{c}{$0.8847$} & \multicolumn{2}{c}{$0.8570$} & \multicolumn{2}{c}{$0.7813$} \\
\midrule
seeds 0/0 & $-0.0000$ & $[-0.0012, +0.0012]$ & $-0.0007$ & $[-0.0016, +0.0001]$ & $+0.0063$ & $[+0.0034, +0.0094]$ \\
seeds 0/1 & $+0.0006$ & $[-0.0004, +0.0016]$ & $-0.0003$ & $[-0.0012, +0.0005]$ & $+0.0050$ & $[+0.0015, +0.0086]$ \\
seeds 1/0 & $+0.0006$ & $[-0.0003, +0.0017]$ & $-0.0019$ & $[-0.0030, -0.0009]$ & $+0.0088$ & $[+0.0056, +0.0122]$ \\
seeds 1/1 & $+0.0008$ & $[-0.0000, +0.0017]$ & $+0.0015$ & $[+0.0005, +0.0024]$ & $+0.0039$ & $[+0.0006, +0.0072]$ \\
ensemble / RNN 0 & $+0.0005$ & $[-0.0004, +0.0015]$ & $+0.0002$ & $[-0.0004, +0.0008]$ & $+0.0076$ & $[+0.0049, +0.0105]$ \\
ensemble / RNN 1 & $+0.0007$ & $[-0.0000, +0.0014]$ & $+0.0014$ & $[+0.0009, +0.0019]$ & $+0.0048$ & $[+0.0019, +0.0077]$ \\
\bottomrule
\end{tabular}

%% file: tables/real_fannie_macro.tex
%
\begin{tabular}{lcccccc}
\toprule
 & \multicolumn{2}{c}{IID} & \multicolumn{2}{c}{Vintage OOD (crisis)} & \multicolumn{2}{c}{Short-history} \\
\cmidrule(lr){2-3}\cmidrule(lr){4-5}\cmidrule(lr){6-7}
Contrast / seeds (XGB/RNN) & $\Delta$AUC & 95\% CI & $\Delta$AUC & 95\% CI & $\Delta$AUC & 95\% CI \\
\midrule
\multicolumn{7}{l}{\emph{macro bank $-$ frame (primary)}} \\
seeds 0/0 & $+0.0141$ & $[+0.0105, +0.0180]$ & $-0.0106$ & $[-0.0143, -0.0069]$ & $+0.0547$ & $[+0.0462, +0.0630]$ \\
seeds 0/1 & $+0.0143$ & $[+0.0108, +0.0181]$ & $+0.0134$ & $[+0.0113, +0.0156]$ & $+0.0548$ & $[+0.0461, +0.0630]$ \\
seeds 0/2 & $+0.0137$ & $[+0.0100, +0.0176]$ & $-0.0057$ & $[-0.0094, -0.0018]$ & $+0.0571$ & $[+0.0482, +0.0655]$ \\
seeds 0/3 & $+0.0144$ & $[+0.0109, +0.0182]$ & $+0.0152$ & $[+0.0132, +0.0173]$ & $+0.0553$ & $[+0.0467, +0.0636]$ \\
seeds 0/4 & $+0.0139$ & $[+0.0102, +0.0176]$ & $+0.0002$ & $[-0.0029, +0.0035]$ & $+0.0582$ & $[+0.0494, +0.0664]$ \\
seeds 1/0 & $+0.0134$ & $[+0.0097, +0.0173]$ & $+0.0184$ & $[+0.0162, +0.0206]$ & $+0.0567$ & $[+0.0486, +0.0645]$ \\
seeds 1/1 & $+0.0140$ & $[+0.0103, +0.0178]$ & $+0.0125$ & $[+0.0105, +0.0145]$ & $+0.0554$ & $[+0.0469, +0.0633]$ \\
seeds 1/2 & $+0.0133$ & $[+0.0096, +0.0174]$ & $+0.0149$ & $[+0.0132, +0.0167]$ & $+0.0561$ & $[+0.0477, +0.0640]$ \\
seeds 1/3 & $+0.0140$ & $[+0.0104, +0.0178]$ & $-0.0098$ & $[-0.0139, -0.0058]$ & $+0.0564$ & $[+0.0482, +0.0642]$ \\
seeds 1/4 & $+0.0139$ & $[+0.0103, +0.0177]$ & $+0.0125$ & $[+0.0111, +0.0139]$ & $+0.0597$ & $[+0.0513, +0.0674]$ \\
ensemble / RNN 0 & $+0.0136$ & $[+0.0100, +0.0174]$ & $+0.0213$ & $[+0.0188, +0.0240]$ & $+0.0560$ & $[+0.0476, +0.0639]$ \\
ensemble / RNN 1 & $+0.0143$ & $[+0.0108, +0.0181]$ & $+0.0156$ & $[+0.0135, +0.0178]$ & $+0.0557$ & $[+0.0473, +0.0634]$ \\
ensemble / RNN 2 & $+0.0137$ & $[+0.0101, +0.0177]$ & $+0.0162$ & $[+0.0137, +0.0188]$ & $+0.0574$ & $[+0.0488, +0.0654]$ \\
ensemble / RNN 3 & $+0.0143$ & $[+0.0108, +0.0180]$ & $+0.0170$ & $[+0.0146, +0.0194]$ & $+0.0570$ & $[+0.0486, +0.0649]$ \\
ensemble / RNN 4 & $+0.0141$ & $[+0.0106, +0.0178]$ & $+0.0152$ & $[+0.0132, +0.0172]$ & $+0.0599$ & $[+0.0517, +0.0677]$ \\
\midrule
\multicolumn{7}{l}{\emph{macro bank $-$ static bank}} \\
ensemble / RNN 0 & $+0.0134$ & $[+0.0097, +0.0172]$ & $+0.0198$ & $[+0.0172, +0.0222]$ & $+0.0479$ & $[+0.0395, +0.0560]$ \\
ensemble / RNN 1 & $+0.0139$ & $[+0.0103, +0.0176]$ & $+0.0153$ & $[+0.0133, +0.0175]$ & $+0.0472$ & $[+0.0387, +0.0553]$ \\
ensemble / RNN 2 & $+0.0128$ & $[+0.0091, +0.0168]$ & $+0.0158$ & $[+0.0133, +0.0184]$ & $+0.0504$ & $[+0.0421, +0.0582]$ \\
ensemble / RNN 3 & $+0.0134$ & $[+0.0097, +0.0172]$ & $+0.0177$ & $[+0.0153, +0.0201]$ & $+0.0484$ & $[+0.0400, +0.0563]$ \\
ensemble / RNN 4 & $+0.0139$ & $[+0.0103, +0.0177]$ & $+0.0149$ & $[+0.0129, +0.0168]$ & $+0.0510$ & $[+0.0429, +0.0591]$ \\
\midrule
\multicolumn{7}{l}{\emph{macro bank on top of joint}} \\
ensemble / RNN 0 & $+0.0131$ & $[+0.0096, +0.0169]$ & $+0.0161$ & $[+0.0145, +0.0177]$ & $+0.0514$ & $[+0.0429, +0.0596]$ \\
ensemble / RNN 1 & $+0.0139$ & $[+0.0104, +0.0175]$ & $+0.0121$ & $[+0.0106, +0.0137]$ & $+0.0540$ & $[+0.0455, +0.0617]$ \\
ensemble / RNN 2 & $+0.0134$ & $[+0.0097, +0.0172]$ & $+0.0154$ & $[+0.0132, +0.0176]$ & $+0.0544$ & $[+0.0457, +0.0627]$ \\
ensemble / RNN 3 & $+0.0135$ & $[+0.0100, +0.0171]$ & $+0.0146$ & $[+0.0133, +0.0159]$ & $+0.0502$ & $[+0.0414, +0.0588]$ \\
ensemble / RNN 4 & $+0.0134$ & $[+0.0099, +0.0172]$ & $+0.0086$ & $[+0.0075, +0.0096]$ & $+0.0543$ & $[+0.0456, +0.0624]$ \\
\bottomrule
\end{tabular}

%% file: tables/real_fannie_bands.tex
%
\begin{tabular}{lccccc}
\toprule
 &  & \multicolumn{2}{c}{macro $-$ frame} & \multicolumn{2}{c}{macro $-$ static} \\
\cmidrule(lr){3-4}\cmidrule(lr){5-6}
Cell & RNN & $\Delta$AUC & 95\% CI & $\Delta$AUC & 95\% CI \\
\midrule
modern vintage (2015/2017) & 0 & $+0.0055$ & $[+0.0035, +0.0075]$ & $+0.0030$ & $[+0.0011, +0.0050]$ \\
 & 1 & $+0.0048$ & $[+0.0025, +0.0070]$ & $+0.0031$ & $[+0.0010, +0.0054]$ \\
modern iid & 0 & $+0.0545$ & $[+0.0418, +0.0661]$ & $+0.0550$ & $[+0.0425, +0.0675]$ \\
 & 1 & $+0.0530$ & $[+0.0397, +0.0654]$ & $+0.0527$ & $[+0.0394, +0.0652]$ \\
cutoff 2004 & 0 & $+0.0086$ & $[+0.0070, +0.0103]$ & $+0.0082$ & $[+0.0066, +0.0099]$ \\
 & 1 & $+0.0042$ & $[+0.0035, +0.0049]$ & $+0.0013$ & $[+0.0007, +0.0019]$ \\
cutoff 2005 (anchor) & 0 & $+0.0213$ & $[+0.0188, +0.0240]$ & $+0.0198$ & $[+0.0172, +0.0222]$ \\
 & 1 & $+0.0156$ & $[+0.0135, +0.0178]$ & $+0.0153$ & $[+0.0133, +0.0175]$ \\
cutoff 2006 & 0 & $+0.0093$ & $[+0.0077, +0.0108]$ & $+0.0092$ & $[+0.0076, +0.0107]$ \\
 & 1 & $+0.0118$ & $[+0.0100, +0.0137]$ & $+0.0123$ & $[+0.0104, +0.0141]$ \\
\bottomrule
\end{tabular}

%% file: tables/real_freddie.tex
%
\begin{tabular}{lcccccc}
\toprule
 & \multicolumn{2}{c}{IID} & \multicolumn{2}{c}{Vintage OOD (crisis)} & \multicolumn{2}{c}{Short-history} \\
\cmidrule(lr){2-3}\cmidrule(lr){4-5}\cmidrule(lr){6-7}
Contrast / RNN seed & $\Delta$AUC & 95\% CI & $\Delta$AUC & 95\% CI & $\Delta$AUC & 95\% CI \\
\midrule
\multicolumn{7}{l}{\emph{macro bank $-$ frame (primary)}} \\
RNN 0 & $+0.0134$ & $[+0.0097, +0.0174]$ & $+0.0119$ & $[+0.0094, +0.0144]$ & $+0.0428$ & $[+0.0334, +0.0522]$ \\
RNN 1 & $+0.0134$ & $[+0.0097, +0.0175]$ & $+0.0100$ & $[+0.0076, +0.0124]$ & $+0.0411$ & $[+0.0309, +0.0507]$ \\
RNN 2 & $+0.0134$ & $[+0.0097, +0.0173]$ & $+0.0103$ & $[+0.0080, +0.0127]$ & $+0.0461$ & $[+0.0364, +0.0553]$ \\
RNN 3 & $+0.0135$ & $[+0.0100, +0.0172]$ & $+0.0062$ & $[+0.0035, +0.0090]$ & $+0.0435$ & $[+0.0338, +0.0527]$ \\
RNN 4 & $+0.0138$ & $[+0.0101, +0.0178]$ & $+0.0127$ & $[+0.0100, +0.0156]$ & $+0.0445$ & $[+0.0349, +0.0535]$ \\
\midrule
\multicolumn{7}{l}{\emph{macro bank $-$ static bank}} \\
RNN 0 & $+0.0111$ & $[+0.0075, +0.0148]$ & $+0.0098$ & $[+0.0074, +0.0124]$ & $+0.0242$ & $[+0.0160, +0.0324]$ \\
RNN 1 & $+0.0112$ & $[+0.0077, +0.0149]$ & $+0.0081$ & $[+0.0059, +0.0106]$ & $+0.0216$ & $[+0.0131, +0.0301]$ \\
RNN 2 & $+0.0109$ & $[+0.0072, +0.0146]$ & $+0.0083$ & $[+0.0061, +0.0105]$ & $+0.0279$ & $[+0.0198, +0.0363]$ \\
RNN 3 & $+0.0117$ & $[+0.0082, +0.0152]$ & $+0.0048$ & $[+0.0020, +0.0076]$ & $+0.0199$ & $[+0.0118, +0.0279]$ \\
RNN 4 & $+0.0146$ & $[+0.0109, +0.0185]$ & $+0.0101$ & $[+0.0074, +0.0129]$ & $+0.0193$ & $[+0.0115, +0.0270]$ \\
\midrule
\multicolumn{7}{l}{\emph{macro bank replaces frame}} \\
RNN 0 & $+0.0089$ & $[+0.0037, +0.0145]$ & $-0.0274$ & $[-0.0312, -0.0236]$ & $+0.0873$ & $[+0.0679, +0.1060]$ \\
RNN 1 & $+0.0082$ & $[+0.0029, +0.0139]$ & $-0.0009$ & $[-0.0043, +0.0025]$ & $+0.0883$ & $[+0.0691, +0.1074]$ \\
RNN 2 & $+0.0087$ & $[+0.0033, +0.0143]$ & $-0.0002$ & $[-0.0033, +0.0031]$ & $+0.0901$ & $[+0.0703, +0.1093]$ \\
RNN 3 & $+0.0087$ & $[+0.0035, +0.0144]$ & $-0.0006$ & $[-0.0036, +0.0024]$ & $+0.0889$ & $[+0.0689, +0.1082]$ \\
RNN 4 & $+0.0085$ & $[+0.0030, +0.0143]$ & $+0.0068$ & $[+0.0038, +0.0098]$ & $+0.0908$ & $[+0.0707, +0.1096]$ \\
\midrule
\multicolumn{7}{l}{\emph{static bank replaces frame (control)}} \\
RNN 0 & $-0.0160$ & $[-0.0216, -0.0107]$ & $-0.0950$ & $[-0.1000, -0.0900]$ & $+0.0749$ & $[+0.0575, +0.0917]$ \\
RNN 1 & $-0.0177$ & $[-0.0236, -0.0119]$ & $-0.0747$ & $[-0.0790, -0.0706]$ & $+0.0717$ & $[+0.0548, +0.0882]$ \\
RNN 2 & $-0.0163$ & $[-0.0222, -0.0109]$ & $-0.0247$ & $[-0.0273, -0.0224]$ & $+0.0727$ & $[+0.0554, +0.0894]$ \\
RNN 3 & $-0.0163$ & $[-0.0220, -0.0110]$ & $-0.0255$ & $[-0.0281, -0.0232]$ & $+0.0737$ & $[+0.0564, +0.0901]$ \\
RNN 4 & $-0.0211$ & $[-0.0271, -0.0154]$ & $-0.0995$ & $[-0.1051, -0.0940]$ & $+0.0733$ & $[+0.0560, +0.0901]$ \\
\midrule
\multicolumn{7}{l}{\emph{primary, 5-channel arm}} \\
RNN 0 & $+0.0135$ & $[+0.0099, +0.0174]$ & $+0.0093$ & $[+0.0073, +0.0114]$ & $+0.0438$ & $[+0.0340, +0.0534]$ \\
RNN 1 & $+0.0134$ & $[+0.0096, +0.0174]$ & $+0.0079$ & $[+0.0056, +0.0102]$ & $+0.0419$ & $[+0.0318, +0.0516]$ \\
RNN 2 & $+0.0133$ & $[+0.0097, +0.0173]$ & $+0.0100$ & $[+0.0079, +0.0121]$ & $+0.0461$ & $[+0.0359, +0.0555]$ \\
RNN 3 & $+0.0135$ & $[+0.0100, +0.0173]$ & $+0.0049$ & $[+0.0026, +0.0075]$ & $+0.0455$ & $[+0.0356, +0.0550]$ \\
RNN 4 & $+0.0138$ & $[+0.0101, +0.0179]$ & $+0.0114$ & $[+0.0089, +0.0139]$ & $+0.0449$ & $[+0.0351, +0.0541]$ \\
\midrule
\multicolumn{7}{l}{\emph{conditioning, 5-channel arm}} \\
RNN 0 & $+0.0111$ & $[+0.0075, +0.0148]$ & $+0.0092$ & $[+0.0072, +0.0113]$ & $+0.0238$ & $[+0.0153, +0.0323]$ \\
RNN 1 & $+0.0111$ & $[+0.0074, +0.0148]$ & $+0.0079$ & $[+0.0057, +0.0103]$ & $+0.0157$ & $[+0.0071, +0.0243]$ \\
RNN 2 & $+0.0109$ & $[+0.0072, +0.0147]$ & $+0.0098$ & $[+0.0076, +0.0121]$ & $+0.0250$ & $[+0.0166, +0.0336]$ \\
RNN 3 & $+0.0114$ & $[+0.0080, +0.0150]$ & $+0.0038$ & $[+0.0014, +0.0063]$ & $+0.0184$ & $[+0.0103, +0.0269]$ \\
RNN 4 & $+0.0146$ & $[+0.0109, +0.0185]$ & $+0.0106$ & $[+0.0081, +0.0131]$ & $+0.0160$ & $[+0.0081, +0.0238]$ \\
\bottomrule
\end{tabular}

%% file: tables/real_fannie_pandemic.tex
%
\begin{tabular}{lcccccc}
\toprule
 & \multicolumn{2}{c}{IID} & \multicolumn{2}{c}{Vintage OOD (pandemic)} & \multicolumn{2}{c}{Short-history} \\
\cmidrule(lr){2-3}\cmidrule(lr){4-5}\cmidrule(lr){6-7}
Contrast / RNN seed & $\Delta$AUC & 95\% CI & $\Delta$AUC & 95\% CI & $\Delta$AUC & 95\% CI \\
\midrule
\multicolumn{7}{l}{\emph{Fannie Mae: macro bank $-$ frame (primary)}} \\
RNN 0 & $+0.0070$ & $[+0.0023, +0.0121]$ & $+0.0049$ & $[+0.0029, +0.0070]$ & $+0.0170$ & $[+0.0062, +0.0286]$ \\
RNN 1 & $+0.0066$ & $[+0.0021, +0.0115]$ & $+0.0049$ & $[+0.0028, +0.0070]$ & $+0.0173$ & $[+0.0068, +0.0291]$ \\
RNN 2 & $+0.0068$ & $[+0.0025, +0.0113]$ & $+0.0050$ & $[+0.0029, +0.0071]$ & $+0.0164$ & $[+0.0040, +0.0299]$ \\
RNN 3 & $+0.0068$ & $[+0.0026, +0.0117]$ & $+0.0044$ & $[+0.0026, +0.0063]$ & $+0.0187$ & $[+0.0075, +0.0307]$ \\
RNN 4 & $+0.0080$ & $[+0.0030, +0.0135]$ & $+0.0047$ & $[+0.0027, +0.0068]$ & $+0.0190$ & $[+0.0084, +0.0313]$ \\
\midrule
\multicolumn{7}{l}{\emph{Fannie Mae: macro bank $-$ static bank}} \\
RNN 0 & $+0.0065$ & $[+0.0015, +0.0118]$ & $+0.0040$ & $[+0.0021, +0.0060]$ & $+0.0139$ & $[+0.0020, +0.0268]$ \\
RNN 1 & $+0.0061$ & $[+0.0013, +0.0112]$ & $+0.0039$ & $[+0.0018, +0.0061]$ & $+0.0129$ & $[+0.0019, +0.0260]$ \\
RNN 2 & $+0.0066$ & $[+0.0024, +0.0113]$ & $+0.0041$ & $[+0.0022, +0.0061]$ & $+0.0106$ & $[-0.0032, +0.0261]$ \\
RNN 3 & $+0.0062$ & $[+0.0017, +0.0111]$ & $+0.0037$ & $[+0.0017, +0.0057]$ & $+0.0143$ & $[+0.0028, +0.0275]$ \\
RNN 4 & $+0.0079$ & $[+0.0027, +0.0137]$ & $+0.0036$ & $[+0.0015, +0.0058]$ & $+0.0151$ & $[+0.0029, +0.0290]$ \\
\midrule
\multicolumn{7}{l}{\emph{Fannie Mae: joint-3 additive (secondary)}} \\
RNN 0 & $+0.0068$ & $[+0.0024, +0.0118]$ & $+0.0041$ & $[+0.0023, +0.0060]$ & $+0.0122$ & $[+0.0005, +0.0247]$ \\
RNN 1 & $+0.0049$ & $[+0.0008, +0.0096]$ & $+0.0034$ & $[+0.0015, +0.0055]$ & $+0.0111$ & $[-0.0030, +0.0269]$ \\
RNN 2 & $+0.0052$ & $[+0.0013, +0.0095]$ & $+0.0028$ & $[+0.0006, +0.0052]$ & $+0.0119$ & $[+0.0000, +0.0256]$ \\
RNN 3 & $+0.0053$ & $[+0.0009, +0.0106]$ & $+0.0030$ & $[+0.0010, +0.0049]$ & $+0.0148$ & $[+0.0013, +0.0300]$ \\
RNN 4 & $+0.0054$ & $[+0.0011, +0.0104]$ & $+0.0036$ & $[+0.0014, +0.0059]$ & $+0.0117$ & $[-0.0013, +0.0267]$ \\
\midrule
\multicolumn{7}{l}{\emph{Freddie Mac: macro bank $-$ frame (primary)}} \\
RNN 0 & $+0.0050$ & $[-0.0003, +0.0105]$ & $+0.0009$ & $[-0.0012, +0.0031]$ & $+0.0132$ & $[-0.0068, +0.0370]$ \\
RNN 1 & $+0.0062$ & $[+0.0003, +0.0125]$ & $+0.0024$ & $[-0.0002, +0.0053]$ & $+0.0272$ & $[+0.0048, +0.0541]$ \\
\midrule
\multicolumn{7}{l}{\emph{Freddie Mac: macro bank $-$ static bank}} \\
RNN 0 & $+0.0061$ & $[+0.0004, +0.0117]$ & $+0.0046$ & $[+0.0023, +0.0070]$ & $+0.0129$ & $[-0.0077, +0.0372]$ \\
RNN 1 & $+0.0068$ & $[+0.0010, +0.0132]$ & $+0.0046$ & $[+0.0016, +0.0077]$ & $+0.0346$ & $[+0.0124, +0.0626]$ \\
\midrule
\multicolumn{7}{l}{\emph{Freddie Mac: joint-3 additive (secondary)}} \\
RNN 0 & $+0.0061$ & $[-0.0013, +0.0142]$ & $+0.0019$ & $[-0.0016, +0.0055]$ & $+0.0689$ & $[+0.0262, +0.1182]$ \\
RNN 1 & $+0.0045$ & $[-0.0005, +0.0107]$ & $+0.0085$ & $[+0.0056, +0.0116]$ & $+0.0777$ & $[+0.0377, +0.1207]$ \\
\bottomrule
\end{tabular}

%% file: tables/real_seed_meta.tex
%
\begin{tabular}{llcccc}
\toprule
Contrast & Regime & Pooled $\Delta$AUC & 95\% CI & $\hat\tau$ & Per-seed range \\
\midrule
Fannie: macro bank $-$ frame (primary) & IID & $+0.0140$ & $[+0.0123, +0.0156]$ & $0.0000$ & $[+0.0136, +0.0143]$ \\
 & Vintage OOD (crisis) & $+0.0170$ & $[+0.0139, +0.0200]$ & $0.0020$ & $[+0.0152, +0.0213]$ \\
 & Short-history & $+0.0572$ & $[+0.0536, +0.0608]$ & $0.0000$ & $[+0.0557, +0.0599]$ \\
\midrule
Fannie: macro bank $-$ static bank & IID & $+0.0135$ & $[+0.0118, +0.0152]$ & $0.0000$ & $[+0.0128, +0.0139]$ \\
 & Vintage OOD (crisis) & $+0.0166$ & $[+0.0141, +0.0191]$ & $0.0016$ & $[+0.0149, +0.0198]$ \\
 & Short-history & $+0.0490$ & $[+0.0454, +0.0527]$ & $0.0000$ & $[+0.0472, +0.0510]$ \\
\midrule
Freddie: macro bank $-$ frame (primary) & IID & $+0.0135$ & $[+0.0118, +0.0152]$ & $0.0000$ & $[+0.0134, +0.0138]$ \\
 & Vintage OOD (crisis) & $+0.0102$ & $[+0.0072, +0.0133]$ & $0.0019$ & $[+0.0062, +0.0127]$ \\
 & Short-history & $+0.0436$ & $[+0.0394, +0.0479]$ & $0.0000$ & $[+0.0411, +0.0461]$ \\
\midrule
Fannie pandemic: macro bank $-$ frame (primary) & IID & $+0.0070$ & $[+0.0049, +0.0091]$ & $0.0000$ & $[+0.0066, +0.0080]$ \\
 & Vintage OOD (pandemic) & $+0.0048$ & $[+0.0039, +0.0057]$ & $0.0000$ & $[+0.0044, +0.0050]$ \\
 & Short-history & $+0.0177$ & $[+0.0125, +0.0229]$ & $0.0000$ & $[+0.0164, +0.0190]$ \\
\midrule
Fannie pandemic: macro bank $-$ static bank & IID & $+0.0066$ & $[+0.0044, +0.0088]$ & $0.0000$ & $[+0.0061, +0.0079]$ \\
 & Vintage OOD (pandemic) & $+0.0039$ & $[+0.0030, +0.0048]$ & $0.0000$ & $[+0.0036, +0.0041]$ \\
 & Short-history & $+0.0135$ & $[+0.0078, +0.0192]$ & $0.0000$ & $[+0.0106, +0.0151]$ \\
\bottomrule
\end{tabular}

%% file: tables/hc_stability_ladder.tex
%
\begin{tabular}{lcccccc}
\toprule
 & \multicolumn{3}{c}{Time (late weeks)} & \multicolumn{3}{c}{IID} \\
\cmidrule(lr){2-4}\cmidrule(lr){5-7}
Frame & Base & $\Delta$stab. & 95\% CI & Base & $\Delta$stab. & 95\% CI \\
\midrule
\multicolumn{7}{l}{\emph{Stability metric}} \\
Free & $0.4272$ & $+0.0259$ & $[-0.0569, +0.1276]$ & $0.5531$ & $-0.0020$ & $[-0.0259, +0.0219]$ \\
Confident & $0.3740$ & $+0.0270$ & $[-0.0449, +0.1267]$ & $0.5474$ & $+0.0007$ & $[-0.0193, +0.0219]$ \\
Strict & $0.3883$ & $-0.0259$ & $[-0.1141, +0.0660]$ & $0.5351$ & $-0.0102$ & $[-0.0374, +0.0016]$ \\
Summaries & $-0.0869$ & $+0.2192$ & $[-0.1757, +0.5678]$ & $0.1944$ & $+0.0920$ & $[+0.0534, +0.1329]$ \\
\midrule
\multicolumn{7}{l}{\emph{Raw AUC}} \\
Free & $0.8306$ & $-0.0016$ & $[-0.0052, +0.0017]$ & $0.8141$ & $-0.0010$ & $[-0.0033, +0.0012]$ \\
Confident & $0.8264$ & $-0.0041$ & $[-0.0073, -0.0009]$ & $0.8092$ & $+0.0002$ & $[-0.0020, +0.0024]$ \\
Strict & $0.8210$ & $-0.0052$ & $[-0.0086, -0.0018]$ & $0.8014$ & $+0.0001$ & $[-0.0019, +0.0022]$ \\
Summaries & $0.6944$ & $+0.0127$ & $[-0.0009, +0.0262]$ & $0.6459$ & $+0.0407$ & $[+0.0304, +0.0511]$ \\
\midrule
\multicolumn{7}{l}{\emph{Stability metric, 5-seed XGB ensemble (per RNN seed)}} \\
Free (RNN 0) & $0.4425$ & $+0.0137$ & $[-0.0452, +0.0785]$ & $0.5606$ & $-0.0012$ & $[-0.0189, +0.0203]$ \\
Confident (RNN 0) & $0.4204$ & $+0.0134$ & $[-0.0381, +0.0824]$ & $0.5510$ & $-0.0037$ & $[-0.0211, +0.0070]$ \\
Strict (RNN 0) & $0.4046$ & $-0.0086$ & $[-0.0765, +0.0615]$ & $0.5354$ & $-0.0059$ & $[-0.0248, +0.0033]$ \\
Summaries (RNN 0) & $-0.1161$ & $+0.2573$ & $[-0.1273, +0.6141]$ & $0.2046$ & $+0.0666$ & $[+0.0267, +0.1099]$ \\
Free (RNN 1) & $0.4425$ & $+0.0307$ & $[-0.0208, +0.0897]$ & $0.5606$ & $-0.0050$ & $[-0.0206, +0.0108]$ \\
Confident (RNN 1) & $0.4204$ & $+0.0113$ & $[-0.0432, +0.0688]$ & $0.5510$ & $-0.0023$ & $[-0.0127, +0.0118]$ \\
Strict (RNN 1) & $0.4046$ & $+0.0050$ & $[-0.0413, +0.0589]$ & $0.5354$ & $-0.0073$ & $[-0.0241, +0.0014]$ \\
Summaries (RNN 1) & $-0.1161$ & $+0.2944$ & $[-0.0350, +0.6413]$ & $0.2046$ & $+0.0620$ & $[+0.0267, +0.1001]$ \\
\bottomrule
\end{tabular}

%% file: tables/real_fannie_econ.tex
%
\begin{tabular}{lccc}
\toprule
Metric at approval rate & 80\% & 90\% & 95\% \\
\midrule
Approved default rate, frame (\%) & $1.60$ & $2.25$ & $2.81$ \\
Approved default rate, $+$macro (\%) & $1.39$--$1.50$ & $2.19$--$2.26$ & $2.76$--$2.83$ \\
Default capture in declined set, frame (\%) & $71.0$ & $54.0$ & $39.4$ \\
Default capture in declined set, $+$macro (\%) & $72.7$--$74.8$ & $53.9$--$55.2$ & $39.0$--$40.4$ \\
Swap set (\% of cohort) & $3.4$--$4.8$ & $1.7$--$2.3$ & $0.9$--$1.0$ \\
Swap-in default rate (\%) & $3.3$--$4.4$ & $8.4$--$9.3$ & $11.9$--$15.8$ \\
Swap-out default rate (\%) & $5.5$--$7.2$ & $9.1$--$10.8$ & $14.1$--$16.7$ \\
Net defaults avoided per 10k approved & $9.2$--$20.9$ & $-0.4$ to $+5.8$ & $-1.6$ to $+4.9$ \\
Defaulted-UPB bps avoided in approved book & $10$--$27$ & $-3$ to $+9$ & $-1$ to $+7$ \\
\bottomrule
\end{tabular}

%% file: main.bbl
\begin{thebibliography}{10}
\expandafter\ifx\csname url\endcsname\relax
  \def\url#1{\texttt{#1}}\fi
\expandafter\ifx\csname urlprefix\endcsname\relax\def\urlprefix{URL }\fi
\expandafter\ifx\csname href\endcsname\relax
  \def\href#1#2{#2} \def\path#1{#1}\fi

\bibitem{fed2011_sr11_7}
{SR 11-7}: Guidance on model risk management, Board of Governors of the Federal
  Reserve System Supervisory Letter, april 4, 2011 (2011).

\bibitem{occ2011_12}
Occ bulletin 2011-12: Sound practices for model risk management: Supervisory
  guidance on model risk management, Office of the Comptroller of the Currency
  Bulletin, april 4, 2011 (2011).

\bibitem{experian_trended3d}
{Experian Information Solutions},
  \href{https://www.experian.com/content/dam/marketing/na/assets/im/decision-analytics/product-sheets/trended-3d-product-sheet.pdf}{Trended
  3d: Trended attributes synthesizing 24 months of credit history}, Product
  sheet (2020).
\newline\urlprefix\url{https://www.experian.com/content/dam/marketing/na/assets/im/decision-analytics/product-sheets/trended-3d-product-sheet.pdf}

\bibitem{vantagescore4_userguide}
{VantageScore Solutions}, \href{https://www.vantagescore.com}{Vantagescore 4.0
  user guide}, Technical documentation (2022).
\newline\urlprefix\url{https://www.vantagescore.com}

\bibitem{chen2016xgboost}
T.~Chen, C.~Guestrin, Xgboost: A scalable tree boosting system, in: Proceedings
  of the 22nd ACM SIGKDD International Conference on Knowledge Discovery and
  Data Mining, 2016, pp. 785--794.

\bibitem{ke2017lightgbm}
G.~Ke, Q.~Meng, T.~Finley, T.~Wang, W.~Chen, W.~Ma, Q.~Ye, T.-Y. Liu, Lightgbm:
  A highly efficient gradient boosting decision tree, in: Advances in Neural
  Information Processing Systems, 2017, pp. 3146--3154.

\bibitem{lightgbm_monotone}
Lightgbm documentation: monotone constraints, LightGBM Parameters documentation
  (\texttt{monotone\_constraints}), accessed 2026-01-29 (2021).

\bibitem{instrument2026}
Y.~L. Tan, Measuring learned monotone temporal aggregation at matched
  admissibility, arXiv preprint arXiv:2610.05196 (2026).

\bibitem{thomas2017credit_scoring}
L.~C. Thomas, J.~N. Crook, D.~B. Edelman, Credit Scoring and Its Applications,
  2nd Edition, Society for Industrial and Applied Mathematics, Philadelphia,
  2017.

\bibitem{crook2007consumer_credit}
J.~N. Crook, D.~B. Edelman, L.~C. Thomas, Recent developments in consumer
  credit risk assessment, European Journal of Operational Research 183~(3)
  (2007) 1447--1465.

\bibitem{archer1993monotonic}
N.~P. Archer, S.~Wang, Application of the back propagation neural network
  algorithm with monotonicity constraints for two-group classification
  problems, Decision Sciences 24~(1) (1993) 60--75.

\bibitem{sill1997monotonic}
J.~Sill, Monotonic networks, in: Advances in Neural Information Processing
  Systems, 1997.

\bibitem{gupta2016lattice}
M.~Gupta, A.~Cotter, J.~Pfeifer, K.~Voevodski, K.~Canini, A.~Mangylov,
  W.~Moczydlowski, A.~van Esbroeck, Monotonic calibrated interpolated look-up
  tables, Journal of Machine Learning Research 17~(109) (2016) 1--47.

\bibitem{you2017deep}
S.~You, D.~Ding, K.~Canini, J.~Pfeifer, M.~R. Gupta, Deep lattice networks and
  partial monotonic functions, in: Advances in Neural Information Processing
  Systems, 2017.

\bibitem{wehenkel2019umnn}
A.~Wehenkel, G.~Louppe, Unconstrained monotonic neural networks, in: Advances
  in Neural Information Processing Systems, 2019.

\bibitem{liu2020certified_monotonic}
X.~Liu, X.~Han, N.~Zhang, Q.~Liu, Certified monotonic neural networks, in:
  Advances in Neural Information Processing Systems, Vol.~33, 2020, pp.
  15427--15438.

\bibitem{runje2023cmnn}
D.~Runje, S.~M. Shankaranarayana, Constrained monotonic neural networks, in:
  Proceedings of the 40th International Conference on Machine Learning, Vol.
  202 of Proceedings of Machine Learning Research, 2023, pp. 29338--29353.

\bibitem{chen2022mnam}
D.~Chen, W.~Ye, Monotonic neural additive models: Pursuing regulated machine
  learning models for credit scoring, in: Proceedings of the Third ACM
  International Conference on AI in Finance, 2022, pp. 70--78.

\bibitem{chen2023monotonicity_mrm}
D.~Chen, W.~Ye, How to address monotonicity for model risk management?, in:
  Proceedings of the 40th International Conference on Machine Learning, 2023,
  pp. 5282--5295.

\bibitem{boardman2024mrnn}
J.~Boardman, X.~Huang, Machine-learning techniques involving monotonic
  recurrent neural networks, U.S. Patent 11{,}960{,}993, assigned to Equifax
  Inc., filed 2020-11-10; granted 2024-04-16 (2024).

\bibitem{kvamme2018mortgage}
H.~Kvamme, N.~Sellereite, K.~Aas, S.~Sjursen, Predicting mortgage default using
  convolutional neural networks, Expert Systems with Applications 102 (2018)
  207--217.
\newblock \href {https://doi.org/10.1016/j.eswa.2018.02.029}
  {\path{doi:10.1016/j.eswa.2018.02.029}}.

\bibitem{babaev2019etrnn}
D.~Babaev, M.~Savchenko, A.~Tuzhilin, D.~Umerenkov, E.t.-{RNN}: Applying deep
  learning to credit loan applications, in: Proceedings of the 25th ACM SIGKDD
  International Conference on Knowledge Discovery \& Data Mining, 2019, pp.
  2183--2190.

\bibitem{gunnarsson2021deep}
B.~R. Gunnarsson, S.~vanden Broucke, B.~Baesens, M.~{\'O}skarsd{\'o}ttir,
  W.~Lemahieu, Deep learning for credit scoring: Do or don't?, European Journal
  of Operational Research 295~(1) (2021) 292--305.
\newblock \href {https://doi.org/10.1016/j.ejor.2021.03.006}
  {\path{doi:10.1016/j.ejor.2021.03.006}}.

\bibitem{hyndman2002ets}
R.~J. Hyndman, A.~B. Koehler, R.~D. Snyder, S.~Grose, A state space framework
  for automatic forecasting using exponential smoothing methods, International
  Journal of Forecasting 18~(3) (2002) 439--454.
\newblock \href {https://doi.org/10.1016/S0169-2070(01)00110-8}
  {\path{doi:10.1016/S0169-2070(01)00110-8}}.

\bibitem{riskmetrics1996}
{J.P. Morgan/Reuters}, {RiskMetrics}---technical document, Tech. rep., J.P.
  Morgan/Reuters, New York, fourth edition (1996).

\bibitem{goetzmann2003hwm}
W.~N. Goetzmann, J.~E. Ingersoll, S.~A. Ross, High-water marks and hedge fund
  management contracts, The Journal of Finance 58~(4) (2003) 1685--1718.

\bibitem{gal2016dropout}
Y.~Gal, Z.~Ghahramani, Dropout as a bayesian approximation: Representing model
  uncertainty in deep learning, in: Proceedings of the 33rd International
  Conference on Machine Learning, Vol.~48 of Proceedings of Machine Learning
  Research, 2016, pp. 1050--1059.

\bibitem{daw2019pga_uq}
A.~Daw, R.~Q. Thomas, C.~C. Carey, J.~S. Read, A.~P. Appling, A.~Karpatne,
  Physics-guided architecture (pga) of neural networks for quantifying
  uncertainty in lake temperature modeling, in: Proceedings of the 2020 SIAM
  International Conference on Data Mining (SDM), 2020, pp. 532--540.
\newblock \href {https://doi.org/10.1137/1.9781611976236.60}
  {\path{doi:10.1137/1.9781611976236.60}}.

\bibitem{wilson1998portfolio}
T.~C. Wilson, Portfolio credit risk, Federal Reserve Bank of New York Economic
  Policy Review 4~(3) (1998) 71--82.

\bibitem{vasicek2002distribution}
O.~A. Vasicek, Loan portfolio value, Risk 15~(12) (2002) 160--162.

\bibitem{bellotti2009credit}
T.~Bellotti, J.~Crook, Credit scoring with macroeconomic variables using
  survival analysis, Journal of the Operational Research Society 60~(12) (2009)
  1699--1707.
\newblock \href {https://doi.org/10.1057/jors.2008.130}
  {\path{doi:10.1057/jors.2008.130}}.

\bibitem{deng2000mortgage}
Y.~Deng, J.~M. Quigley, R.~Van~Order, Mortgage terminations, heterogeneity and
  the exercise of mortgage options, Econometrica 68~(2) (2000) 275--307.
\newblock \href {https://doi.org/10.1111/1468-0262.00110}
  {\path{doi:10.1111/1468-0262.00110}}.

\bibitem{aguais2004pit}
S.~D. Aguais, L.~R. Forest, E.~Y.~L. Wong, D.~Diaz-Ledezma, Point-in-time
  versus through-the-cycle ratings, in: M.~K. Ong (Ed.), The Basel Handbook: A
  Guide for Financial Practitioners, Risk Books, London, 2004.

\bibitem{krueger2018provisioning}
S.~Kr{\"u}ger, D.~R{\"o}sch, H.~Scheule, The impact of loan loss provisioning
  on bank capital requirements, Journal of Financial Stability 36 (2018)
  114--129.
\newblock \href {https://doi.org/10.1016/j.jfs.2018.02.009}
  {\path{doi:10.1016/j.jfs.2018.02.009}}.

\bibitem{bellotti2013stress}
T.~Bellotti, J.~Crook, Forecasting and stress testing credit card default using
  dynamic models, International Journal of Forecasting 29~(4) (2013) 563--574.
\newblock \href {https://doi.org/10.1016/j.ijforecast.2013.04.003}
  {\path{doi:10.1016/j.ijforecast.2013.04.003}}.

\end{thebibliography}
